\documentclass[11pt]{article}

\usepackage{thesis_style}

\title{Model-free trading and hedging with continuous price paths}

\author{Tigran Atoyan}

\begin{document}

%this baselineskip gives sufficient line spacing for an examiner to easily
%markup the thesis with comments
\baselineskip=18pt plus1pt

%set the number of sectioning levels that get number and appear in the contents
\setcounter{secnumdepth}{3}
\setcounter{tocdepth}{3}

\maketitle

% abstract
\begin{abstract}
In this paper, we provide a model-independent extension of the paradigm of dynamic hedging of derivative claims. We relate model-independent replication strategies to local martingales having a closed form which we can characterise via solutions of coupled PDEs. We provide a general framework and then apply it to a market with no traded claims, a market with an underlying asset and a convex claim and a market with an underlying asset and a set of co-maturing call options. The results encompass known examples of model-independent identities and provide a methodology for deriving new identities.
\end{abstract}

% introduction
\section{Introduction}
\label{section: introduction}

In a landmark year for the financial derivatives industry, Black, Scholes and Merton \cite{Black-Scholes, Merton} revolutionised option pricing by introducing the paradigm of dynamic hedging. Notably, they showed that under some market assumptions, the financial risk of issuing an option can be fully hedged by the continuous-time limit of a dynamic trading strategy in the underlying asset. The gains and losses incurred by the strategy match the fluctuations of the derived price of the option, which is guaranteed to converge to the option payoff. This methodology eliminated the need to estimate the drift of the asset price, a task which is known to be notoriously difficult (Rogers \cite{Rogers} and Monoyios \cite{Monoyios}). Due to its remarkable simplicity and practicality,  the Black-Scholes model was quickly embraced by practitioners and has been the pillar of option pricing and hedging thereafter. 

A key limitation of the Black-Scholes model and its many variants is the requirement to make assumptions on the volatility of the underlying asset. This paper provides a framework which tackles this issue by providing an extension of the original dynamic hedging paradigm. In particular, we characterise a set of payoffs one may replicate through continuous time dynamic trading in a frictionless market with continuous prices when one does not specify a model for the asset prices. This yields a general methodology which encompasses most (if not all) of the model-independent replication strategies proposed in the literature. There are several advantages of the proposed approach, provided one can solve systems of linear parabolic PDEs. First, it yields a set of "marketable" claims which have model-independent prices under all continuous semimartingales. These replication results in fact characterise the full set of model-independent strategies with wealth processes of a certain form. Second, the model-independent prices are enforced by hedging strategies with hedge ratios given in closed form. And thirdly, the proposed framework allows for model-independent dynamic hedging strategies with traded claims.

In Section \ref{section: motivating examples}, we first provide some known examples of model-independent results to motivate the general replication results in Section \ref{section: replication theory}, which are further concretised in Section \ref{section: examples}. Sections \ref{section: characterisation theory} - \ref{section: traded calls} addresses the question of characterising the full set of replication strategies with portfolios of a given functional form. In particular, Section \ref{section: characterisation theory} provides a general result in abstract form which is then applied in Sections \ref{section: convex claim} - \ref{section: examples with convex claim} for a market with an underlying asset and a convex claim and in Section \ref{section: traded calls} for a market with an underlying asset and a set of traded calls. Section \ref{section: conclusion} provides summary remarks and the Appendix contains some technical proofs.

% notation
\section*{Notation}
\label{section: notation}

Throughout this paper, we will consider financial markets with a riskless asset of constant value and a set of risky assets. The assets may be traded in continuous time and without transaction costs. Asset prices will typically be denoted by $A$, but this convention changes in certains sections where more suitable notation is applicable. The following summarises the basic notation used in the thesis.

\begin{notation}[General]

For $x, y \in \R^k$, $xy$ denotes their inner product. For matrices $x, y \in \R^{k \times k}$, $xy$ denotes their matrix product and $x \cdot y$ denotes their Hadamard product, notably, $x \cdot y := \sum_{i,j}^k x_{ij} y_{ij}$. For $x \in \R^k$ and $y \in \R^l$, denote $(x,y) \equiv (x_1,\ldots,x_k,y_1,\ldots,y_l)$. For points $x, y \in \R^k$, $d(x,y)$ denotes their Euclidean distance. For a point $x \in \R^k$ and a set $K \subset \R^k$, define $d(x,K) := \inf \{ d(x,y) : y \in K \}$. For $x \in \R^k$ and $r > 0$, define the closed ball of radius $r$ around $x$ by $B^c_r(x) := \{ y \in \R^k: d(x,y) \leq r \}$. For a set $K \subseteq \R^k$, denote $cl(K)$ and $K^c$ to be respectively the closure and the complement of $K$. Denote $\partial K := cl(K) \cap cl(K^c)$ to be the boundary set of $K$. 

Define $\R_+ := [0, \infty)$ and denote the space of symmetric positive definite $\R^{k \times k}$ matrices by $\mathbb{S}^k_+$. For a function $f: \R^k \to \R$,  denote its gradient and Hessian respectively by $\nabla f$ and $D^2 f$. For $k \geq 1$, $T > 0$ and functions $f, g: [0,T] \to \R^k$,
$$
\int_0^t f(u) \ud g(u) := \sum_{i = 1}^k \int_0^t f_i (u) \ud g_i (u), \quad t \in [0,T],
$$
when the integrals  $\int_0^t f_i (u) \ud g_i (u)$ are well-defined. Similarly, for $f, g: [0,T] \to \R^{k \times k}$,
$$
\int_0^t f(u) \cdot \ud g(u) = \sum_{i,j = 1}^k \int_0^t f_{ij}(u) \ud g_{ij}(u), \quad t \in [0,T].
$$
when the integrals $\int_0^t f_{ij}(u) \ud g_{ij}(u)$ are well-defined.

\end{notation}

\begin{notation}[Stochastic analysis]

For a filtered probability space $(\Omega, \cF, \F_T, \Prob)$, where $\F_T = \{ \cF_t \}$ can be any filtration with $\cF_s \subseteq \cF_t \subseteq \cF$ for $0 \leq s \leq t \leq T$, we denote the $i$-th component of an $\F_T$-adapted process $P: \Omega \times [0,T] \to \R^k$ by $P^i$ and say that $P$ is a (continuous) semimartingale on $(\Omega, \cF, \F_T, \Prob)$ if each of its components $P^i$, $i \in \{1, \ldots, k\}$ is an $\F_T$-adapted (continuous) semimartingale. We sometimes set $T = \infty$, in which case we use the shorthand notation $\F := \{\cF_t\}_{t \geq 0}$. Similar notation and definitions apply for matrix-valued semimartingales. For semimartingales $Y: \Omega \times [0, T] \to \R^{d_1 \times d_2}$ and $P: \Omega \times [0, T] \to \R^{d_2}$ on $(\Omega, \cF, \F_T, \Prob)$, define
$$
\int_0^t Y_u \ud P_u(\omega) := \left( \sum_{j = 1}^{d_2} \int_0^t Y^{1,j}_u \ud P^j_u (\omega), \ldots, \sum_{j = 1}^{d_2} \int_0^t Y^{d_1,j}_u \ud P^j_u (\omega) \right), \quad (\omega, t) \in \Omega \times [0,T],
$$
where the $\int_0^t Y^{ij}_u \ud P^j_u (\omega)$ are It\^o integrals. For a semimartingale $P: \Omega \times [0,T] \to \R^k$ on $(\Omega, \cF, \F_T, \Prob)$, denote sample paths of $P$ by $P(\omega) := \{ P(\omega, t) \}_{t \in [0,T]}$, $\omega \in \Omega$.

For two filtered probability spaces $(\Omega^1, \cF^1, \F^1_T, \Prob^1)$ and $(\Omega^2, \cF^2, \F^2_T, \Prob^2)$, define their product space by
$$
(\Omega^1, \cF^1, \F^1_T, \Prob^1) \otimes (\Omega^2, \cF^2, \F^2_T, \Prob^2) := (\Omega^1 \times \Omega^2, \cF^1 \otimes \cF^2, \F^1_T \otimes \F^2_T, \Prob^1 \otimes \Prob^2),
$$
where $\F^1_T \otimes \F^2_T := \{ \cF^1_t \otimes \cF^2_t \}_{t \in [0,T]}$. Denote the canonical space of an $\R^k$-valued Brownian motion on $[0,T]$ by $(\cC(\R^k), \cF^W_T, \F^W_T,\Prob^W)$.

A filtered probability space $(\widetilde{\Omega},\widetilde{\cF},\widetilde{\F}_T,\widetilde{\Prob})$ contains $(\Omega, \cF, \F_T, \Prob)$ if $\exists (\Omega^\prime, \cF^\prime, \F^\prime_T, \Prob^\prime)$ such that $(\widetilde{\Omega},\widetilde{\cF},\widetilde{\F}_T,\widetilde{\Prob}) = (\Omega, \cF, \F_T, \Prob) \otimes (\Omega^\prime, \cF^\prime, \F^\prime_T, \Prob^\prime)$.

Denote the set of stopping times with respect to a filtration $\F_T$ by $\cT(\F_T)$. For stopping times $\tau_1, \tau_2 \in \cT(\F_T)$ on a filtered space $(\Omega, \cF, \F_T, \Prob)$, define the notation
\begin{align*}
&\{ \tau_1 < \tau_2 \} := \{ \omega \in \Omega : \tau_1(\omega) < \tau_2(\omega) \} \in \cF, \\
&(\tau_1, \tau_2) := \{ (\omega, t) \in \Omega \times [0,t] : \tau_1(\omega) < t < \tau_2(\omega) \} \in \cF \otimes \cB([0,T]),
\end{align*}
where $\cB([0,T])$ is the Borel $\sigma$-algebra on $[0,T]$. $\{\tau_1 \leq \tau_2 \}$, $[\tau_1, \tau_2]$, $(\tau_1, \tau_2]$ and $[\tau_1, \tau_2)$ are defined analogously.

For an $\R$-valued continuous semimartingale $Z$, denote its local time at $a \in \R$ on $[0,t]$ by $L^a_t[Z]$ and its Dol\'eans-Dade exponential $e^{Z - \frac{1}{2} \langle Z \rangle}$ by $\cE(Z)$.

\end{notation}

% motivating examples
\section{Motivating examples}
\label{section: motivating examples}

To motivate the results of this paper and provide some preliminary intuition, we start with two basic examples, due to Bick \cite{Bick} and Carr and Lee \cite{Carr-Lee:semimartingales}, of the kind of strategies we will be presenting in a more general setting. Within this section, we consider a market with an underlying asset $S$ with continuous price paths and we set the time horizon to be $[0,\infty)$.
 
In \cite{Bick}, Bick showed that it is possible to replicate path-independent claims on $S$ with maturity equal to hitting times of realised variance. In particular, the replication strategy holds under all continuous semimartingale models for $S$. The argument proving the result is remarkably simple.

For $q > 0$ and a payoff function $f$, by which we henceforth mean the function defining the value of a claim at maturity, consider the solution $F$ to the PDE
\begin{equation}
\label{BS PDE}
\frac{\partial}{\partial x_2}F(x_1,x_2) + \frac{1}{2} x_1^{2} \frac{\partial^2}{\partial x_1^2} F(x_1,x_2) = 0, \quad F(x_1,q)=f(x_1).
\end{equation}
Note that $F(x_1,x_2)$ is equal to the Black-Scholes price of the path-independent claim with payoff function $f$ under zero interest rate and with $x_1$ and $x_2$ respectively representing the observed asset price and the realised variance to be accumulated until maturity (volatility squared multiplied by time to maturity). For $q > 0$, define $\tau$ to be the first time when $\log S$ accumulates realised variance $q$, notably $\tau := \inf\{t>0: \langle \log S \rangle_{t} = q\}$. It\^o's formula applied to $F(S, \langle \log S \rangle)$, along with
the finite variation property of $\langle \log S \rangle$ gives
\begin{align*}
F(S_{t},\langle \log S \rangle_{t}) & =  F(S_{0},0) + \int_{0}^{t} \frac{\partial}{\partial x_1} F(S_{u},\langle \log S \rangle_{u}) \ud
S_u \\
&\quad \quad + \int_{0}^{t} \frac{\partial}{\partial x_2} F(S_{u},\langle \log S \rangle_{u})\ud \langle \log S \rangle_{u} \\
& \quad \quad + \frac{1}{2} \int_0^t \frac{\partial^2}{\partial x_1^2} F(S_u,\langle \log S \rangle_u) \ud \langle S \rangle_u, \quad t \in [0,\tau] \setminus \{ \infty \}.
\end{align*}
Now using the PDE \eqref{BS PDE} satisfied by $F$ along with
$$
\ud \langle \log S \rangle_u = \frac{1}{S_u^2} \ud \langle S \rangle_u,
$$
we obtain the pathwise identity
\begin{equation}
\label{example integral representation 1}
F(S_{t},\langle \log S \rangle_{t}) = F(S_{0},0) + \int_{0}^{t} \frac{\partial}{\partial x_1} F(S_{u},\langle \log S \rangle_{u})\ud
S_{u}, \quad [0,\tau] \setminus \{ \infty \}.   
\end{equation}
The interpretation of this equation is that by judicious choice of the
inputs to $F$, and using the PDE satisfied by $F$, we have succeeded in
mimicking the situation whereby an agent starts with initial capital
$F(S_{0},0)$ at time zero, uses the self-financing trading strategy which holds
$H := \frac{\partial}{\partial x_1} F(S,\langle \log S \rangle)$ units of $S$, and at
all times $t \in [0,\tau] \setminus \{ \infty \}$, her wealth is equal to $F(S_{t},\langle \log S \rangle_{t})$. Equivalently, we trade with a delta-hedging rule
generated by the hedging function $F$, and the resulting gain from trade, as given by the integral
$\int_{0}^{\cdot} \frac{\partial}{\partial x_1} F(S_t,\langle \log S \rangle_t)\ud S_{u}$, is the only term left in
the integral form of the process
$F(S_{t},\langle \log S \rangle_{\cdot \wedge \tau}) - F(S_0,0)$.

Now consider a claim with payoff $f(S_\tau)$ and maturity at time $\tau$, when the quadratic variation of $\log S$ reaches $q$. Timer calls and puts are examples of such options which have been traded in practice (see Carr and Lee \cite{Carr-Lee:volatility}). Since $\langle \log S \rangle_\tau=q$, \eqref{example integral representation 1} evaluated at $t = \tau$, along with the initial condition $F(x,q)=f(x)$ satisfied by $F$ results in
\begin{equation}
\label{example integral representation 2}
f(S_{\tau}) = F(S_{0},0) + \int_{0}^{\tau} \frac{\partial}{\partial x_1} F(S_{t},\langle \log S \rangle_{t})\ud S_{t}.
\end{equation}
for $\tau < \infty$. Thus, by combining a solution $F$ to a parabolic PDE with a suitable choice of the inputs to $F$, we have generated a perfect hedge for the European timer claim using a delta-hedging rule, as given by \eqref{example integral representation 2}, under all continuous semimartingale models where $\tau < \infty$ almost surely, provided the possibility to trade the asset frictionlessly in continuous time.

\begin{proposition}[Bick \cite{Bick}]
\label{proposition: Bick}
For functions $f: \R_+  \to \R$ and $w: \R_+^2 \to \R_+$ respectively representing a payoff function and a weighting function, and for $q > 0$, let $F$ satisfy the PDE 
\begin{equation}
\label{Bick PDE}
2 \frac{\partial F}{\partial x_2} + x^2 \frac{\partial^2 F}{\partial x_1^2} = 0, \quad F(x_1,q)=f(x_1).
\end{equation}
For $q > 0$, define the stopping time $\tau \equiv \tau_q$ by $\tau_q := \inf \{t > 0: \langle \log S \rangle_t = q\}$. Then,
\begin{equation}
\label{Bick integral representation}
F(S_{\cdot \wedge \tau},Q^w_{\cdot \wedge \tau}) = F(S_{0},0) + \int_{0}^{\cdot \wedge \tau} \frac{\partial}{\partial x_1} F(S_{u},\langle \log S \rangle_u) \ud S_{u}
\end{equation}
under all models where $S$ is a continuous semimartingale. In particular, the self-financing strategy with initial wealth $F(S_{0},0)$ and holding $H := \frac{\partial}{\partial x_1} F(S,\langle \log S \rangle)$ units of $S$ replicates the claim with payoff $f(S_{\tau})$ at maturity $\tau$ under all continuous semimartingale models where $\tau < \infty$ almost surely.
\end{proposition}

Similarly, Propositions 2.9 and 2.10 of Carr and Lee \cite{Carr-Lee:semimartingales} propose model-independent replication strategies for claims on realised variance with maturity equal to the hitting time of price to a barrier or the exit time of price from an interval. We present a version of their results below.

\begin{proposition}[Carr and Lee \cite{Carr-Lee:semimartingales}]
\label{proposition: Carr and Lee}
For functions $f: \R_+ \to \R_+$ be a payoff function, and for $l < S_0 < u$, let $F$ satisfy the PDE
\begin{equation}
\label{Carr-Lee PDE}
2 \frac{\partial F}{\partial x_2} + x^2 \frac{\partial^2 F}{\partial x_1^2} = 0, \quad F(l, x_2) = F(u, x_2)=f(x_2).
\end{equation}
For $q > 0$, define the stopping time $\tau \equiv \tau_{l,u}$ by $\tau_{l,u} := \inf \{t\geq0: S_t \not\in (l,u)\}$. Then,
\begin{equation}
\label{Carr-Lee integral representation}
F(S_{\cdot \wedge \tau},\langle \log S \rangle^w_{\cdot \wedge \tau}) = F(S_{0},0) + \int_{0}^{\cdot \wedge \tau} \frac{\partial}{\partial x_1} F(S_{u},\langle \log S \rangle_{u}) \ud S_{u}
\end{equation}
under all models where $S$ is a continuous semimartingale. In particular, the self-financing strategy with initial wealth $F(S_{0},0)$ and holding $H := \frac{\partial}{\partial x_1} F(S,\langle \log S \rangle)$ units of $S$ replicates the claim with payoff $f(\langle \log S \rangle_{\tau})$ at maturity $\tau$ under all continuous semimartingale models where $\tau < \infty$ almost surely.
\end{proposition}

% replicaton
\section{Model-independent replication}
\label{section: replication theory}

We now consider a financial market with a riskless asset and $d$ risky assets with discounted prices $A = (A^1,\ldots,A^d)$. The time horizon is $[0,T]$. We allow for some of the assets to be traded options co-maturing at time $T$. When there are no traded options in the market, we may set $T = \infty$, in which case $[0,T]$ should be interpreted as $[0,\infty)$. We begin by defining what we mean by a model for $A$.

\begin{definition}
\label{definition: continuous semimartingale model}
We say that a filtered probability space $(\Omega, \cF, \F_T, \Prob)$ and a stochastic process $A: \Omega \times [0,T] \to \R^d$ form a \textit{continuous semimartingale (resp. local martingale / martingale) model} $M = \{A; (\Omega, \cF, \F_T, \Prob)\}$ for $A$ if the process $A$ is a continuous semimartingale (resp. local martingale / martingale) on $(\Omega, \cF, \F_T, \Prob)$ and its components corresponding to traded claims attain their payoff at maturity, by which we mean that the components representing traded options must satisfy their payoff conditions almost surely. Denote the family of continuous semimartingale (resp. local martingale / martingale) models for which $A_0 = a$ by $\cM \equiv \cM_s(a)$ (resp. $\cM_\ell \equiv \cM_\ell(a)$ / $\cM_m \equiv \cM_m(a)$).
\end{definition}

Note that $\cM_m \subset \cM_\ell \subset \cM_s$ and that, unless we specify otherwise, we will assume a common starting point $A_0 = a$ when referring to different sets of models. For a non-anticipative functional $X$ of $A$, we will use the shorthand notation $X_\cdot \equiv X_\cdot[A]$.

\begin{definition}
\label{definition: non-anticipative functional}
Let $\cP_1$ and $\cP_2$ be two pathspaces whose paths are defined on $[0,T]$ and are continuous. A map $X: \cP_1 \to \cP_2$ is called a \textit{non-anticipative functional} if $X_t[P] = X_t[P^\prime]$ for any $t \in [0,T]$ and any paths $P, P^\prime \in \cP_1$ which are identical on $[0,t]$.
\end{definition}

If $A$ is a process adapted to a filtration $\F_T$ and $X$ is a non-anticipative functional, then $X_\cdot[A]$ is adapted to $\F_T$. See Cont and Fourni\'e \cite{Cont-Fournie:pathwise, Cont-Fournie:stochastic} for measurability properties of non-anticipative functionals in a more general setting.

\begin{definition}
\label{definition: functionals in cX}
Define $\cX$ to be the set of non-anticipative functionals $X = (X^1,\ldots,X^n)$ such that, for each $i \in \{1,\ldots,n\}$ and under all continuous semimartingale models $M \in \cM_s$, $X^i_\cdot \equiv X^i_\cdot[A]$ has the integral form
\begin{equation}
\label{integral form of X^i}
X^i_t = X^i_0 + \int_0^t \alpha^i(X_u) \ud A_u + \int_0^t \beta^i(X_u) \cdot \ud \langle A \rangle_u + \int_0^t \gamma^i(X_u) \ud u, \quad t \in [0,T]
\end{equation}
for a constant $X_0 = (X^1_0,\ldots,X^n_0)$ and for continuous functions $\alpha^i: \R^n \to \R^d$, $\beta^i: \R^n \to \R^{d \times d}$ and $\gamma^i: \R^n \to \R$.
\end{definition}

For a set of functions $\alpha^i$, $\beta^i$, $\gamma^i$, $i \in \{1,\ldots,n\}$, define the following operators on smooth enough ($C^2$ is sufficient) functions $F: \R^n \to \R$:
\begin{align}
\mathcal{L}^{\alpha, \beta}_{ij} F &:= \sum_{k=1}^n \beta^k_{i,j} \frac{\partial}{\partial x_k} F + \frac{1}{2} \sum_{k,l=1}^n \alpha^k_i \alpha^l_j \frac{\partial^2}{\partial x_k \partial x_l} F , \quad i,j \in \{1,\ldots,d\}
\label{second variation operator} \\
\mathcal{L}^{\gamma} F &:= \sum_{k=1}^n \gamma^k \frac{\partial}{\partial x_k} F
\label{time variation operator} \\
\mathcal{L}^{\alpha} F &:= \sum_{k=1}^n \alpha^k \frac{\partial}{\partial x_k} F
\label{hedge ratio operator}
\end{align}
Note that $\mathcal{L}^{\alpha, \beta}_{i,j} F$, $i,j \in \{1,\ldots,d\}$ and $\mathcal{L}^{\gamma} F$ are functions from $\R^n$ to $\R$, whereas $\mathcal{L}^{\alpha} F$ is a function from $\R^n$ to $\R^d$.

\begin{definition}
\label{definition: solution set to system of PDEs}
For $X \in \cX$ with corresponding functions $\alpha^i, \beta^i, \gamma^i$, $i \in \{1,\ldots,d\}$ and for a connected set $\cD \subseteq \R^n$, we define the family of functions $\cS_X(\cD)$ by
$$
\cS_X(\cD) := \{F \in C^2(\cD):  \cL^{\gamma} F = \cL^{\alpha,\beta}_{i,j} F = 0, \quad i,j \in \{1,\ldots,d\} \}.
$$
\end{definition}

One may relax the $C^2$ assumption if some of the $\alpha^i$ are identically equal to zero. 

\begin{definition}
\label{definition: range of X}
Let $\cM \subseteq \cM_s$ be a subset of models for $A$ and let $X$ be a non-anticipative functional. Then, the \textit{range of $X$ on $\cM$} is defined by 
$$
\cR(X; \cM) := \{ x \in \R^n : \forall \epsilon > 0, \ \exists \{A; (\Omega, \cF, \F_T, \Prob)\} \in \cM \ \ s.t. \ \ \Prob(\tau^X_{B^c_\epsilon(x)} < T) > 0 \},
$$
where $B^c_\epsilon(x)$ is the closed ball of radius $\epsilon$ around $x$.
\end{definition}

The proof of the following proposition is provided in the appendix.

\begin{proposition}
\label{proposition: properties of range of X}
Let $\cM \subseteq \cM_s$ be a non-empty set of models for $A$ and $X$ be a non-anticipative functional. Then, for any starting point $a \in \R^d$ of $A$, the range $\cR(X; \cM(a))$ of $X$ on $\cM(a)$ is a closed and connected subset of $\R^n$, and $X_t[A] \in \cR(X; \cM)$ for all $t \in [0,T)$ almost surely under all models in $\cM(a)$.
\end{proposition}

Henceforth, when referring to integral representation identities which hold under a set of models $M \in \cM$, we will mean that the identities hold almost surely for each $M \in \cM$. 

For a closed set $B$, we define $\tau^X_B$ to be the first hitting time of $X_\cdot$ to $B$. If $X_\cdot$ does not attain $B$ prior to time $T$, set $\tau^X_B = \infty$. Denote $X^B$ to be the functional $X$ stopped when hitting $B$ (hence, $X^B_\cdot \equiv X_{\cdot \wedge \tau_B}$). The following proposition and its corollary, derived using It\^o calculus, provide a general framework for deriving model-free replication results.

\begin{proposition}
\label{proposition: integral representation}
Consider a functional $X \in \cX$ and a closed set $B \subseteq \R^n$. Denote $\cD \equiv \cR(X^B;\cM_s) \setminus B$. Then, for any $F: \cD \cup B \to \R$ in $\cS_X(\cD) \cap \cC(\cD \cup B)$ and under all continuous semimartingale models for $M \in \cM_s$ for $A$:
\item (1) $F(X^B_\cdot)$ admits the integral representation
\begin{equation}
\label{integral representation of F}
F(X^B_\cdot) \equiv F(X_0) + \int_0^\cdot \I_{\{t < \tau^X_B\}} \cL^\alpha F(X_t) \ud A_t
\end{equation}

\item (2) If $T \in (0,\infty)$, $\int_0^T F(X^B_t) \ud t$ admits the integral representation
\begin{equation}
\label{integral representation of cash flow}
\int_0^T F(X^B_t) \ud t = T F(X_0) + \int_0^{T} \I_{\{t < \tau^X_B\}} (T-t) \cL^\alpha F(X_t) \ud A_t.
\end{equation}
\end{proposition}

\begin{proof}

\item Proof of claim(1): 

Let $M = \{A; (\Omega, \cF, \F_T, \Prob)\} \in \cM_s$. By $X \in \cX$, there is a set of functions $\alpha^i$, $\beta^i$ and $\gamma^i$, $i \in \{1\,\ldots,n\}$ such that
$$
\ud X^i_t = \alpha^i(X_t) \ud A_t + \beta^i(X_t) \ud \langle A \rangle_t + \gamma^i(X_t) \ud t
$$
for each $i \in \{1,\ldots,n\}$. Denoting $v \otimes w$ to be the outer product of vectors $v$ and $w$,
\begin{align*}
\ud F(X_t) &= \sum_{k=1}^n \frac{\partial}{\partial x_k} F(X_t) \ud X^k_t + \frac{1}{2} \sum_{k,l=1}^n \frac{\partial^2}{\partial x_k \partial x_l} F(X_t) \ud \langle X^k, X^l \rangle_t \\
&= \sum_{k=1}^n \frac{\partial}{\partial x_k} F(X_t) \alpha^k(X_t) \ud A_t + \sum_{k=1}^n \frac{\partial}{\partial x_k} F(X_t) \beta^k(X_t) \cdot \ud \langle A \rangle_t \\
& \quad + \sum_{k=1}^n \frac{\partial}{\partial x_k} F(X_t) \gamma^k(X_t) \ud t + \frac{1}{2} \sum_{k,l=1}^n \frac{\partial^2}{\partial x_k \partial x_l} F(X_t) ( \alpha^k(X_t) \otimes \alpha^l(X_t) ) \cdot \ud \langle A \rangle_t \\
&= \sum_{k=1}^n \frac{\partial}{\partial x_k} F(X_t) \alpha^k(X_t) \ud A_t + \sum_{k=1}^n \frac{\partial}{\partial x_k} F(X_t) \gamma^k(X_t) \ud t \\
& \quad + \sum_{i,j=1}^d \left( \sum_{k=1}^n \frac{\partial}{\partial x_k} F(X_t) \beta^k_{ij}(X_t) + \frac{1}{2} \sum_{k,l=1}^n \frac{\partial^2}{\partial x_k \partial x_l} F(X_t) \alpha^k_i(X_t) \alpha^l_j (X_t) \right) \ud \langle A^i, A^j \rangle_t
\end{align*}
for $t \in [0, T \wedge \tau^X_B)$. Hence, by definition of the operators in \eqref{second variation operator} -- \eqref{hedge ratio operator}, we get
$$
\ud F(X_t) = \cL^\alpha F(X_t) \ud A_t + \cL^\gamma F(X_t) \ud t + \sum_{i,j=1}^d \cL^{\alpha,\beta}_{i,j} F(X_t) \ud \langle A^i, A^j \rangle_t, \quad t \in [0, T \wedge \tau^X_B).
$$
By the assumption $F \in \cS_X(\cR(X^B; \cM_s) \setminus B)$, notably that
$$
\cL^\gamma F = \cL^{\alpha,\beta}_{i,j} F = 0, \quad i, j \in \{1,\ldots,d\}
$$
on $\cR(X^B; \cM_s) \setminus B$ and by $\Prob(X_t \in \cR(X^B; \cM_s) \setminus B, \ t \in [0,T \wedge \tau^X_B)) = 1$ (which follows from Proposition \ref{proposition: properties of range of X} and by $X_t \not\in B$ on $t \in [0, \tau^X_B)$), we conclude that
$$
\ud F(X_t) = \cL^\alpha F(X_t) \ud A_t, \quad t \in [0, T \wedge \tau^X_B).
$$
Hence, by continuity of $F$ on $B$,
$$
F(X^B_t) = F(X_{t \wedge \tau^X_B}) = F(X_0) + \int_0^t \I_{\{u < \tau^X_B\}} \cL^\alpha F(X_u) \ud A_u, \quad t \in [0,T],
$$
which proves \eqref{integral representation of F}.

\item Proof of claim (2):

Evaluating \eqref{integral representation of F} at $t = T$ and multiplying both sides by $T$, we have
$$
TF(X^B_T) = T F(X_0) + \int_0^T \I_{\{t < \tau^X_B\}} T \cL^\alpha(X_t) \ud A_t.
$$
On the other hand, integration by parts (applied to $t F(X^B_t)$ at $t = T \wedge \tau^X_B$) and $\ud F(X^B_t) = \cL^\alpha F(X^B_t) \ud A_t$ for $t \in [0, T \wedge \tau^X_B)$ give
\begin{align*}
(T \wedge \tau^X_B) F(X^B_T) &= \int_0^{T \wedge \tau^X_B} F(X^B_t) \ud t + \int_0^{T \wedge \tau^X_B} t \ud F(X^B_t) \\
&= \int_0^{T \wedge \tau^X_B} F(X^B_t) \ud t + \int_0^T \I_{\{t < \tau^X_B\}} t \cL^\alpha F(X^B_t) \ud A_t
\end{align*}
By the above identities,
\begin{align*}
&T F(X_0) + \int_0^T \I_{\{t < \tau^B_X\}} (T-t) \cL^\alpha F(X_t) \ud A_t \\
&= \left( T F(X_0) + \int_0^T \I_{\{t < \tau^X_B\}} T \cL^\alpha F(X^B_t) \ud A_t \right) - \int_0^T \I_{\{t < \tau^X_B\}} t \cL^\alpha F(X^B_t) \ud A_t \\
&= T F(X^B_T) - \left( (T \wedge \tau^X_B) F(X^B_T) - \int_0^{T \wedge \tau^X_B} F(X^B_t) \ud t \right) \\
&=  \int_0^T F(X^B_t) \ud t,
\end{align*}
which concludes the proof.

\end{proof}

\begin{corollary}
\label{corollary: model-independent replication}
For $X$ and $B$ as in Proposition \ref{proposition: integral representation}, the following hold under all semimartingale models where $\tau^X_B < \infty$ almost surely:
\item (1) $F(X^B_\cdot)$ is the wealth process of the self-financing strategy starting with capital $F(X_0)$ and holding $\cL^\alpha F(X_t)$ units of $A$ for $t \in [0,T \wedge \tau^X_B)$. Denoting $F = f$ on $B$, this strategy replicates the claim with payoff $f(X_{\tau^X_B})$.
\item (2) The self-financing strategy starting with capital $T F(X_0)$ and holding $(T-t) \cL^\alpha F(X_t)$ units of $A$ for $t \in [0, T \wedge \tau^X_B)$ replicates the cash flow $\int_0^T F(X^B_t) \ud t$.
\end{corollary}

Proposition \ref{proposition: integral representation} and Corollary \ref{corollary: model-independent replication} have probability-free pathwise interpretations. This is presented in much further detail in \cite{Atoyan}.

% examples
\section{Examples with single traded asset}
\label{section: examples}

We now discuss some applications of the results in the previous section, in particular of Corollary \ref{corollary: model-independent replication}. In this section, we consider a market with a single traded underlying asset $A \equiv S$ and time horizon $[0,\infty),$ and consider different choices of functionals $X \in \cX$ of $S$. In particular, we do not consider markets with traded claims in this section (see Section \ref{section: examples with convex claim} for such a market). When referring to model-independent identities, we will mean that these identities apply under all continuous semimartingale models for $S$. Similarly, model-independent replication strategies of claims with maturity $\tau_B$ will refer to strategies which apply under all continuous semimartingale models where the stopping time $\tau_B$ is finite almost surely. In this section, we first highlight basic examples for the case of wealth processes and corresponding claims on price and realised variance, with the latter being of a quite general weighted form. We then discuss claims contingent on running extrema before providing examples with claims having an explicit time dependence.

\subsection{Claims on price and quadratic variation}
\label{subsection: price and quadratic variation}

Let $w: \R \rightarrow \R$ and $Q^w_\cdot = \int_0^\cdot w(S_t) \ud \langle S \rangle_t$, and consider $X := (S,Q^w)$. For a closed set $B \subset \R^2$ and $\cD \equiv \cR(X^B; \cM_s) \setminus B$, $S_X(\cD)$ is the set of $C^{2,1}$ functions which solve the heat equation
\begin{equation}
\label{PDE for price and realised variance}
2 w \frac{\partial}{\partial x_2} F + \frac{\partial^2}{\partial x_1^2} F = 0
\end{equation}
subject to some boundary conditions. Theorem \ref{proposition: integral representation} and Corollary \ref{corollary: model-independent replication} yield the result in Bick \cite{Bick}, notably that for any function $F$ solving \eqref{PDE for price and realised variance}, $F( S_{\cdot \wedge \tau^X_B}, Q^w_{\cdot \wedge \tau^X_B} )$ is the wealth process of the self-financing portfolio starting with capital $F(S_0,0)$ and holding a position $H = \frac{\partial}{\partial x_1} F(S,Q^w)$ in $S$ on $[0,\tau^X_B)$. Different choices of the boundary set $B$ yield different replication results. Setting $B =  \{x_2 = q\}$ for $q > 0$ yields replication of timer claims, as outlined in Section \ref{section: motivating examples}. Setting $B = \{x_1 = l\} \cup \{x_1 = u\}$ for $0 \leq l < S_0 < u \leq \infty$ yields the replication strategies proposed in Carr and Lee \cite{Carr-Lee:semimartingales}. We note one further application with $X = (S, Q^w)$: timer claims with barrier conditions. Let $q > 0$ be the hitting level of $Q^w$ defining the maturity of a knock-out barrier claim with barrier levels $0 \leq l \leq S_0 \leq u \leq \infty$. Temporarily define $\tau_q := \inf\{ t > 0: Q^w_t = q\}$ to be the maturity of the claim and define $\tau_b := \inf\{ t > 0: S_t \not\in (l,u) \}$ to be hitting time of $S$ to either barrier $l$ or $u$. Then, the following corollary follows directly from the integral representation identity \eqref{integral representation of F}.

\begin{corollary}
\label{corollary: barrier claims}
Let $f: \R_+ \to \R$ be such that $f(l) = f(u) = 0$ and that a solution $F$ to the heat equation \eqref{PDE for price and realised variance} and boundary conditions
\begin{align}
&F(x_1, q) = f(x_1), \quad x_1 \in (l,u), \label{boundary condition at maturity} \\
&F(l, x_2) = F(u, x_2) = 0, \quad x_2 \in [0,q], \label{boundary condition at barrier}
\end{align}
exists. Then, the self-financing portfolio starting with capital $F(S_0,0)$ and holding a position $H_t := \frac{\partial}{\partial x_1} F(S_t,Q^w_t)$ in $S$ for $t \in [0,\tau_q \wedge \tau_b)$ replicates the timer claim with payoff $f(S_{\tau_q}) \I_{\{ S_t \in (l, u), \ t \in [0,\tau_q] \}}$ under all continuous semimartingale models where $\tau_q < \infty$ almost surely.
\end{corollary}

\begin{proof}

Setting $B = \{x_1 = l\} \cup \{x_1 = u\} \cup \{x_2 = q\}$ and boundary conditions \eqref{boundary condition at maturity} -- \eqref{boundary condition at barrier} in Corollary \ref{corollary: model-independent replication}, we get that the self-financing strategy starting with capital $F(S_0,0)$ and holding a position $H$ in $S$ until $\tau^X_B = \tau_q \wedge \tau_b$ attains a wealth at time $\tau^X_B$ equal to $0$ if $\tau^X_B = \tau_b$ (in which case a barrier is hit, $H$ is set to $0$ until $\tau_q$ and we trivially replicate the knocked out barrier claim) or to $f(S_{\tau_q})$ if $\tau^X_B = \tau_q$ (the payoff of the timer claim which has not been knocked out).

\end{proof}

Since knock-in barrier claims are the difference between non-barrier claims and knock-out barrier claims, their replication follows from taking the difference of the replication strategies for the latter. 

The replication strategy in Corollary \ref{corollary: barrier claims} only requires inputs $S$ and $Q^w$ to $F$. In particular, the dependence on the extrema is implicitly embedded in the stopping time defining the dynamic hedging period. This is due to the weak path-dependence of single touch barrier claims. The following subsection considers claims which depend explicitly on the running extrema of $S$.

\subsection{Dependence on running extrema}
\label{subsection: dependence on running extrema}

The class $\cX$ does not contain functionals $X$ which include the running maximum or minimum as components. Hence, Theorem \ref{proposition: integral representation} and Corollary \ref{corollary: model-independent replication} cannot be directly applied with such functionals.  We could have extended the class $\cX$ to incorporate dependence on running extrema and generalised Proposition \ref{proposition: integral representation} to this extended class, but we did not do so for the sake of clarity. In this subsection, we discuss such extensions in the context of $A = S$. We also show that the squared drawdown or drawup (difference between price and running minimum) may be included as a component of $X \in \cX$.

\begin{notation}
Within this subsection, denote the running maximum and minimum of $S$ respectively by $M \equiv M^S$ and $m \equiv m^S$.
\end{notation}

\subsubsection{Running extrema as components of $X$} 
The following result shows that one can generate wealth processes which have a model-independent closed form given by functions of $X$, where $X$ includes $M$ as a component.

\begin{proposition}
\label{proposition: claims on S, M and Q^w}
Consider $X = (S, M, Q^w)$ and a closed set $B \subset \R_+^3$ and $\cD \equiv \cR(X^B; \cM_s) \setminus B$. Suppose that $F: \cD \cup B \to \R$ is a $C^{2,1,1}(\cD)$ solution to 
\begin{align}
&2 w \frac{\partial}{\partial x_3} F + \frac{\partial^2}{\partial x_1^2} F = 0, \label{running maximum PDE} \\
& \I_{\{ x_1 = x_2 \}}\frac{\partial}{\partial x_2} F(x_1, x_2, x_3) = 0 \label{mixed BC}
\end{align}
on $\cD$ and is continuous on $B$. Then, $F(X^B_\cdot)$ admits the integral representation
$$
F (X^B_\cdot) = F(X_0) + \int_0^\cdot \I_{\{ t < \tau^X_B \}} \frac{\partial}{\partial x_1} F(X_t) \ud S_t
$$
under all continuous semimartingale models for $S$. Hence, $F(X^B_\cdot)$ is the model-independent wealth process of the self-financing strategy starting with capital $F(X_0)$ and holding $\frac{\partial}{\partial x_1} F(X_t)$ units of $S$ for $t \in [0, \tau^X_B)$. Denoting $F = f$ on $B$, this strategy replicates the claim with payoff $f(X_{\tau^X_B})$ and maturity $\tau^X_B$ under all continuous semimartingale models where $\tau^X_B < \infty$ almost surely. This yields replication of: \\
(a) Claims on $(S, M)$ with maturity equal to hitting times of $Q^w$, either with or without a lower barrier condition on $S$, under all continuous semimartingale models where the hitting times are finite almost surely. \\
(b) Claims on $(M,Q^w)$ with maturity equal to hitting times of $S$ to a level $l < S_0$ under all continuous semimartingale models where the hitting times are finite almost surely.  \\
A similar result applies for $X = (S, m, Q^w)$.
\end{proposition}

\begin{proof}

The results for $X = (S, M, Q^w)$ and $X = (S, m, Q^w)$ follow along the same lines, so we only outline the arguments for $X = (S, M, Q^w)$.

Let $F \in \cC(B;X)$ be a solution to \eqref{running maximum PDE} -- \eqref{mixed BC}. By It\^o's formula (generalised to include the running maximum),
\begin{align*}
\ud F(S_t, M_t, Q^w_t) =  \frac{\partial}{\partial x_1} F \ud S_t +  \frac{\partial}{\partial x_2}F \ud M_t + \left( w \frac{\partial}{\partial x_3}  F + \frac{1}{2}  \frac{\partial^2}{\partial x_1^2}F \right) \ud \langle S \rangle_t, \quad t \in [0,\tau^X_B),
\end{align*}
where we used the shorthand notation $F \equiv F(S_t,M_t,Q^w_t)$. By \eqref{running maximum PDE} -- \eqref{mixed BC}, the continuity of $F$ at $B$ and the observation that $M$ is carried by $\{S = M\}$, it follows that
$$
F(X_t) = F(X_0) + \int_0^t \frac{\partial}{\partial x_1} F(X_u) \ud S_u, \quad t \in [0,\tau^X_B].
$$
under all continuous semimartingale models for $S$. The replication results follow by analogous arguments as in Subsection \ref{subsection: price and quadratic variation}. 
\end{proof}

\subsubsection{Dependence on the squared drawdown and drawup}

This subsection shows that the square of the drawdown process $M - S$ or drawup process $S - m$ may be used as a component of $X$ such that $X \in \cX$. This is due to the following application of well-known results on Az\'ema-Yor martingales (Theorem 1 in Ob\l\'oj \cite{Obloj:Azema-Yor}).  Define $D^M := (M - S)^2$ and $D^m := (S - m)^2$. 

\begin{lemma}
\label{lemma: Azema-Yor}

For all continuous semimartingale models for $S$,
\begin{align}
& D^M_\cdot = S_\cdot^2 - S_0^2 - 2 \int_0^\cdot M_t \ud S_t = \langle S \rangle_\cdot - 2 \int_0^\cdot \sqrt{ D^M_t } \ud S_t,
\label{integral formulation for D^M} \\
& D^m_\cdot = S_\cdot^2 - S_0^2 - 2 \int_0^\cdot m_t \ud S_t = \langle S \rangle_\cdot + 2 \int_0^\cdot  \sqrt{ D^m_t }  \ud S_t,
\label{integral formulation for D^m}
\end{align}

\end{lemma}

\begin{proof}

Let $g$ be a locally bounded function on $\R$ and define $G(x) := \int_0^x g(s) \ud s$. Then, by the finite variation of $g(M_t)$,
$$
\ud \left( g(M_u) (M_u - S_u) \right) =  (g(M_u) \ud M_u - g(M_u) \ud S_u) - (M_u - S_u) \ud M_u
$$
%$$
%g(M_t) (M_t - S_t) - g(M_0) (M_0 - S_0) = \int_0^t g(M_u) \ud (M - S)_u + \int_0^t (M - S)_u \ud g(M_u).
%$$
for $u \geq 0$. Since $M_0 = S_0$, $(M_u - S_u) \ud g(M_u) = 0$, and $\int_0^t g(M_u) \ud M_u = G(M_t) - G(S_0)$, this gives
\begin{align*}
&g(M_t) (M_t - S_t) = G(M_t) - G(S_0) - \int_0^t g(M_u) \ud S_u \\
\Leftrightarrow \ &g(M_t) M_t - G(M_t) - S_t g(M_t) + G(S_0) + \int_0^t g(M_u) \ud S_u = 0.
\end{align*}
The only properties we have used above are that $M$ is carried by $\{S = M\}$, that it has finite variation and that $M_0 = S_0$. Hence, the same arguments applied to $m$ give
$$
g(m_t) m_t - G(m_t) - S_t g(m_t) + G(S_0) + \int_0^t g( m_u ) \ud S_u = 0.
$$
For $g(x) = x$, hence $G(x) = \frac{1}{2} x^2$, this gives
\begin{align*}
&\frac{1}{2} M_t^2 - S_t M_t + \frac{1}{2} S_0^2 + \int_0^t M_u \ud S_u = 0, \quad t \geq 0, \\
&\frac{1}{2} m_t^2 - S_t m_t + \frac{1}{2} S_0^2 + \int_0^t m_u \ud S_u = 0, \quad t \geq 0,
\end{align*}
or, equivalently,
\begin{align*}
& (M_t - S_t)^2 - S_t^2 + S_0^2 + 2 \int_0^t M_u \ud S_u = 0, \quad t \geq 0, \\
& (S_t - m_t)^2 - S_t^2 + S_0^2 + 2 \int_0^t m_u \ud S_u = 0, \quad t \geq 0,
\end{align*}
which rearrange to \eqref{integral formulation for D^M} and \eqref{integral formulation for D^m} respectively by $S^2_\cdot = S_0^2 + 2 \int_0^\cdot S_u \ud S_u + \langle S \rangle_\cdot$.

\end{proof}

\begin{remark}
Let $T > 0$ and consider a market where claims with payoff $(S_T)^2$ at maturity $T$ are traded at price $C$. Then, Lemma \ref{lemma: Azema-Yor} implies that the claim of maturity $T > 0$ with payoff $D^M_T = (M_T - S_T)^2$ is replicated model-independently by the self-financing portfolio starting with capital $C_0 - S_0^2$ and holding a static position of one unit of $C$ along with a dynamic position of $-2 \sqrt{D^M}$ units of $S$. An analogous replication result holds for the claim with payoff $D^m_T$.
\end{remark}

We can apply the representation of $D^M$ and $D^m$ in Lemma \ref{lemma: Azema-Yor} along with Proposition \ref{proposition: integral representation} to obtain the following replication results corresponding to $X = (D^M, \langle S \rangle)$.

\begin{corollary}
\label{corollary: claims on square drawdown and realised variance}
Consider a closed set $B \subset \R^2$, the set $\cD = \cR(X^B; \cM_s) \setminus B$ and a function $f$ on $B$ such that there is a solution $F \in C^{2,1}(\cD) \cap C(B)$ to
\begin{equation}
\label{drawdown PDE}
\frac{\partial}{\partial x_2} F + \frac{\partial}{\partial x_1} F + 2 x_1 \frac{\partial^2}{\partial x_1^2} F = 0.
\end{equation}
with boundary condition $F = f$ on $B$. Then, the self-financing strategy starting with capital $F(0,0)$ and holding $- 2 \sqrt{D^M_t} \frac{\partial}{\partial x_1} F(D^M_t, \langle S \rangle_t)$ units of $S$ for $t \in [0,\tau^X_B)$ replicates the claim with payoff $f( D^M_{\tau^X_B}, \langle S \rangle_{\tau^X_B} )$ and maturity $\tau^X_B$ under all continuous semimartingale models where $\tau^X_B < \infty$ almost surely. This yields replication of: \\
(a) Claims on the drawdown with maturity equal to hitting times of $\langle S \rangle$, either with or without a barrier condition on the drawdown, under all continuous semimartingale models where the hitting times are finite almost surely.  \\
(b) Claims on $Q^w$ with maturity equal to hitting times of the drawdown under all continuous semimartingale models where the hitting times are finite almost surely.

\end{corollary}

The proof follows from Corollary \ref{corollary: model-independent replication} and the same line of arguments as in Subsection \ref{subsection: price and quadratic variation}. A similar result holds for $X = (D^m, \langle S \rangle)$.

In fact, there is an equivalence result which allows capturing dependence on running extrema by functionals in $\cX$ as defined in Definition \ref{definition: functionals in cX} via the drawdown and drawup. This is detailed in Appendix \ref{appendix: equivalence}.

\subsection{Example: timer claims contingent on running extrema}

To illustrate the approaches outlined above, we consider the following example. Suppose we would like to replicate a timer claim with payoff $f(S_\tau, M_\tau)$, where $\tau := \inf\{ t > 0 : \langle S \rangle_t = q\}$ for some $q > 0$. Set $B = \{ x_3 = q\} \subset \R^3$ and $X = (S, M, \langle S \rangle)$. Then, Proposition \ref{proposition: claims on S, M and Q^w} implies that we can replicate the claim under all continuous semimartingale models where $\tau < \infty$ almost surely if, for $\cD \equiv \cR(X^B; \cM_s) \setminus B = \{0 \leq x_1 \leq x_2\} \cap \{0 \leq x_3 < q\} \subset \R^3$, we can solve for a solution $F \in C^{2,1,1}(\cD) \cap C(B)$ to \eqref{running maximum PDE} -- \eqref{mixed BC} with $w = 1$ and boundary condition $F(x_1, x_2, q) = f(x_1, x_2)$.  

Proposition \ref{proposition: stochastic solution for F(S,M,[S])} in the Appendix shows that, provided an implicit regularity condition on the payoff function $f$, the PDE has a stochastic solution given by
\begin{equation}
\label{stochastic solution for F}
F(x_1,x_2,x_3) := \E \left[ f \left( x_1 + W_{q - x_3}, x_2 \vee (x_1 + M^W_{q- x_3}) \right) \right], \quad  0 \leq x_1 \leq x_2, \ \ 0 \leq x_3 \leq q,
\end{equation}
for a Brownian motion $W$ (with running maximum $M^W$) under a probability measure $\Prob$. 

By Proposition \ref{proposition: equivalence between F^M and F^D}, a second approach would be to solve for a $C^{2,2,1}$ solution $F^D$ to the PDE \eqref{PDE for price, squared drawdown and realised variance} with boundary condition $F^D(x_1, x_2, q) = f(x_1, x_1 + \sqrt{x_2})$ and such that $\frac{\partial}{\partial y_2} F^D(y_1, y_2, y_3)$ is well-defined on $\R_+ \times \{0\} \times [0,q)$, and finally set $F := F^D \circ \ytilde$, where $\ytilde$ is defined as in the proposition.

Provided with a solution $F$ by any one of the methods above, the self-financing strategy starting with capital $F(S_0, S_0, 0)$ and holding $\frac{\partial}{\partial x_1} F(S_t, M_t, \langle S \rangle_t)$ units in $S$ for $t \in [0,\tau)$ has wealth process equal to $F(S_{\cdot \wedge \tau}, M_{\cdot \wedge \tau}, \langle S \rangle_{\cdot \wedge \tau})$ and replicates the claim with payoff $f(S_{\tau}, M_{\tau})$ and maturity $\tau$ under all continuous semimartingale models where $\tau < \infty$ almost surely.

\subsection{Time-dependent functionals}
\label{subsection: time dependent functionals}

To conclude this section, we turn to the question of replication of claims of fixed maturity by including time-dependent components in $X$. Whereas the previous examples only involved a single linear parabolic PDE, typically admitting well-defined and unique solutions for a large class of boundary conditions, the inclusion of time-dependent components in $X$ yields a system of two PDEs and restricts the class of payoff functions (corresponding to boundary conditions) for which claims may be replicated in a model-independent manner according to Corollary \ref{corollary: model-independent replication}.

In particular, this section explores the range of claims of fixed maturity $T > 0$ one may replicate via Corollary \ref{corollary: model-independent replication} by considering time as an explicit component of $X$. Setting $X^1 = t$ (corresponding to $\gamma^{1} = 1$, $\alpha^{1} = \beta^{1} = 0$), and $B = \{x_1 = T\} \subset \R^{d}$, we get $\tau^X_B = T$, which by Corollary \ref{corollary: model-independent replication} yields model-independent replication of claims of the form $f(X_T)$, where $f$ is the boundary condition on $B$ for functions in $\cS_X(\cR(X;\cM_s))$.

First, observe that unless there is $i \in \{2,\ldots,n\}$ such that $\gamma^{i}(x) \neq 0$, then $\cL^\gamma F(x) = 0$ implies that $\frac{\partial}{\partial x_1}F = 0$, which would defeat the purpose of including the component $X^1 = t$ in $X$. The most commonly encountered quantity in finance which has an absolutely continuous variation with respect to time is the arithmetic average of the stock price defining Asian payoffs, which modulo scaling is equal to $
V_\cdot \equiv \int_0^\cdot S_t \ud t$.

\begin{remark}
\label{remark: weighted Asian}
As with quadratic variation, we can weight the integrand by a function of $X$, but this would only change the PDE formulations without providing any more insight.
\end{remark}

Setting $X^2 = V$, we need to include $S$ or a bijective function thereof as a component in $X$ in order for $X \in\cX$ to hold (since $\gamma^{2}(X) = S$). As a basic example, consider $X_t = (t, V_t, S_t)$, $t \in [0,T]$. Set $\cD := \cR(X^B; \cM_s)$. It is easy to check that $\cD = [0,T] \times [0,\infty) \times \R_+$.

Corollary \ref{corollary: model-independent replication} implies that $F(X_\cdot)$ is the model-independent wealth process of a self-financing strategy in $S$ if $F \in \cS_X(\cD \setminus B)$, hence if $F$ satisfies the system of PDEs
\begin{align}
&\frac{\partial^2}{\partial x_3^2} F = 0 \label{PDE for F(t,V,S) due to [S]}, \\
&\frac{\partial}{\partial x_1}F + x_3 \frac{\partial}{\partial x_2}F = 0. \label{PDE for F(t,V,S) due to t}
\end{align}
on $\cD \setminus B = [0,T) \times [0,\infty) \times \R_+$.

\begin{proposition}
\label{proposition: solutions for F(t,V,S)}
$F \in C^2(\cD \setminus B) \cap C(\cD)$ is a solution to \eqref{PDE  for F(t,V,S) due to [S]} -- \eqref{PDE for F(t,V,S) due to t} iff
\begin{equation}
\label{F(t,V,S)}
F(x) = c_1 (x_1 x_3 - x_2)  + c_2 x_3 + c_3
\end{equation}
for constants $c_1, c_2, c_3 \in \R$. 
\end{proposition}

\begin{proof}

By \eqref{PDE for F(t,V,S) due to [S]}, there are $C^2$ functions $F_1$ and $F_2$ (temporary notation) such that
$$
F(x) = x_3 F_1(x_1, x_2) + F_2(x_1, x_2)
$$
Applying \eqref{PDE for F(t,V,S) due to t} yields
$$
x_3 \frac{\partial}{\partial x_1} F_1(x_1, x_3) + \frac{\partial}{\partial x_1} F_2(x_1, x_3) + x_3 \left( x_3 \frac{\partial}{\partial x_2} F_1(x_1, x_3) + \frac{\partial}{\partial x_2} F_2(x_1, x_2) \right) = 0,
$$
which may be rearranged to
$$
x_3^2 \frac{\partial}{\partial x_2} F_1(x_1, x_2) + x_3 \left( \frac{\partial}{\partial x_1} F_1(x_1, x_2) + \frac{\partial}{\partial x_2} F_2(x_1, x_2) \right) + \frac{\partial}{\partial x_1} F_2(x_1, x_2) = 0
$$

Since the above relation applies to all $x \in [0,T) \times \R_+^2$, the factors of the different powers of $x_3$ must be zero. Hence, the above is equivalent to 
\[
\begin{cases}
\frac{\partial}{\partial x_2} F_1(x_1, x_2) = 0 \\
\frac{\partial}{\partial x_1} F_1(x_1, x_2) + \frac{\partial}{\partial x_2} F_2(x_1, x_2) = 0 \\
\frac{\partial}{\partial x_1} F_2(x_1, x_2) = 0
\end{cases}
\]
for $(x_1, x_2) \in [0,T) \times \R_+$. This, in turn, is equivalent to
\[
\begin{cases}
F_1(x_1, x_2) = c_1 x_1 + c_2 \\
F_2(x_1, x_2) = -c_1 x_2 + c_3
\end{cases}
\]
for $(x_1, x_2) \in [0,T) \times \R_+$ and constants $c_1, c_2, c_3 \in \R$. This proves that $F \in \cS_X(\cD \setminus B)$ iff $F(x) = x_3 (c_1 x_1 + c_2) + (-c_1 x_2 + c_3)$, which rearranges to \eqref{F(t,V,S)}.

\end{proof}

Hence, for $X_t = (t, V_t, S_t)$, the only functions $F$ for which Corollary \ref{corollary: model-independent replication} implies that $F(X_\cdot)$ is a model-independent wealth process of a self-financing strategy are of the form \eqref{F(t,V,S)}. In particular, this only yields linear combinations of the well-known model-independent replication strategy for the linear Asian payoff -- an application of the integration by parts identity $T S_T = \int_0^T t \ud S_t + \int_0^T S_t \ud t$ -- and of static positions in $S$ and in bonds. 

The above result illustrates the limitations implied by the system of two PDEs corresponding to functionals which depend explicitly on time or, more generally, on integrals with respect to time. To obtain non-trivial model-independent replication strategies with time-dependent functionals, we need to augment $X$ with additional components. For instance, we can set $X_t = (t, V_t, S_t, \langle S \rangle_t)$ and $B = \{x_1 = T\} \subset \R^4$. Then, Corollary \ref{corollary: model-independent replication} implies that $F(X_\cdot)$ is the model-independent wealth process of a self-financing trading strategy if $F$ solves the system of PDEs
\begin{align}
&\frac{\partial}{\partial x_4} F + \frac{1}{2} \frac{\partial^2}{\partial x_3^2} F = 0 \label{PDE for F(t,V,S,[S]) due to [S]}, \\
&\frac{\partial}{\partial x_1} F + x_3 \frac{\partial}{\partial x_2} F = 0 \label{PDE for F(t,V,S,[S]) due to t}.
\end{align}
on $\cR(X^B;\cM_s) \setminus B = [0,T) \times \R_+^3$. Provided this system of PDEs admits a solution with boundary condition $F(T, x_2, x_3, x_4) = f(x_2,x_3,x_4)$, one can replicate $f(V_T, S_T, \langle S \rangle_T)$ under all continuous semimartingale models for $S$.

\begin{remark}[DDS]
As outlined by Carr and Lee \cite{Carr-Lee:semimartingales}, there is an interpretation of the hedging strategies proposed in Bick \cite{Bick} and in \cite{Carr-Lee:semimartingales} based on a Dambis/Dubins-Schwarz time change. This interpretation, outlined below, applies to model-independent replication strategies for claims which are contingent on a single asset (or on a single self-financing portfolio, also denoted an attainable process in the language of Fukasawa \cite{Fukasawa}) and which have no explicit dependence on time.

The DDS interpretation goes as follows. By well-known results (\cite{Delbaen-Schachermayer, Harrison:1981}) in the theory of no arbitrage, the asset price $S$ is a local martingale under a risk-neutral measure provided that the market admits "No Free Lunch With Vanishing Risk". Fix a risk-neutral measure $\Q$. Hence, when $A = S$ is continuous and one-dimensional, the Dambis/Dubins-Schwarz (DDS) theorem implies that $S = W_{\langle S \rangle}$, where $W$ is a Brownian motion under $\Q$. The quadratic variation based hedging strategies presented so far in this section use the observed quadratic variation $\langle S \rangle$ to reduce the problem of hedging a claim for a general continuous local martingale $S$ to hedging a claim for Brownian motion. However, this time change impacts the maturity of the claim. In particular, the maturities are specified on the time clock of the underlying Brownian motion, meaning that they correspond to hitting times of $\langle S \rangle$. More generally, suppose that one of the components of $X$, say $X^n$, is $Q^w_\cdot = \int_0^\cdot w(S_t, Q^w_t) \ud \langle S \rangle_t$, the boundary set $B$ is $\{x_n = q\}$ and the payoff function is $f(x_1,\ldots,x_{n-1})$. If the claim on $f(X^1_q,\ldots,X^{n-1}_q)$ with maturity $q$ admits a pricing function under the model 
$$
\ud S_t = \left( w(S_u, Q^w_u) \right)^{-1/2} \ud W_t, \quad t \in [0,q),
$$
where $W$ is a Brownian motion under the risk-neutral measure, then this pricing function is in the solution set $\cS_X(\cD)$ for a corresponding domain $\cD$ and yields a model-independent replication strategy for the claim on $f(X^1_{\tau^X_B},\ldots,X^{n-1}_{\tau^X_B})$ with maturity $\tau^X_B$ given by the hitting time of $Q^w$ to $q$. Setting $w(x_1,x_2) = \frac{1}{x_1^2}$ yields model-independent strategies based on the Black-Scholes formula, with the time change defined by $\langle \log S \rangle$.

On the other hand, when the claim depends explicitly on $S$ and on time, as with the examples in this section, then the above time change argument generally does not yield analogous model-independent strategies. It only does so under specific constraints on the payoff functions, which translate to constraints on the boundary conditions under which the system of PDEs corresponding to $X$ has a solution. As will be highlighted by examples in Section \ref{section: examples with convex claim}, similar restrictions on payoff functions apply when $A$ is multidimensional.
\end{remark}

% characterisation
\section{Characterisation in a general market setting}
\label{section: characterisation theory}

The results in the previous sections provided sufficient conditions for deriving model-independent hedging strategies. We now show that the conditions we have derived on regular enough functions $F$ such that $F(X)$ is a model-independent wealth process for a given $X \in \cX$ are not only sufficient but also necessary. To do so, we characterise the set of $C^2$ functions $F$ such that $F(X^B_\cdot[A])$ is a local martingale for $X \in \cX$ when $A$ is a generic continuous local martingale. This work has parallels with the results in Ob{\l}oj and Yor \cite{Obloj-Yor}, where the authors characterise the set of local martingales which are functions of a continuous local martingale (with unspecified volatility) and its supremum. One of the main differences is that we have a general specification of $X$ and $B$ (the results we present also apply without any stopping set $B$, with no significant changes to the proof). In particular, throughout this section, $X$ is a $\cC_T(\R^n)$-valued functional in $\cX$ with corresponding functions $\alpha^i$, $\beta^i$, $\gamma^i$, $i \in \{1,\ldots,n\}$ and $B$ is a closed subset of $\R^n$. Another important difference is that we consider multivariate local martingales $A$ which may be assumed to converge to specified payoffs (as would be the case in a market with traded options). 

We begin by noting that for $\cD_{\cM_{\ell}} := \cR(X^B;\cM_{\ell})$ and $F \in C(\cD_{\cM_{\ell}}) \cap C^2(\cD_{\cM_{\ell}} \setminus B)$, Proposition \ref{proposition: integral representation} implies that
$$
F(X^B_\cdot) = F(X_0) + \int_0^\cdot \I_{\{ t < \tau^X_B \}} \cL^\alpha F(X^B_t) \ud A_t
$$
under all models $M \in \cM_{\ell}$. Since $\cL^\alpha F(X^B_\cdot)$ is locally bounded (by continuity) and since $A$ is by definition a local martingale under all $M \in \cM_{\ell}$, it follows that $F(X^B_\cdot)$ is a local martingale under all $M \in \cM_{\ell}$.

In this section, we show that the characterisation result holds for families of models which satisfy a certain regularity property. In the following sections, we discuss some important cases of markets represented by $A$ and solve for families of models which satisfy the latter regularity conditions for these markets.

\begin{definition}
\label{definition: regular class of models}
We say that a family of local martingale models $\cM \subseteq \cM_{\ell}$ for $A$ is \textit{regular} if for any model $\{ A; (\Omega, \cF, \F_T, \Prob)\} \in \cM$, matrix $\Sigma \in \mathbb{S}^d_+$ and stopping time $\tau \in \cT(\F_T)$, there is a model $\{\widetilde{A}; (\widetilde{\Omega},\widetilde{\cF},\widetilde{\F}_T,\widetilde{\Prob})\} \in \cM$ and a stopping time $\widetilde{\tau} \in \cT(\widetilde{\F}_T)$ such that:
\begin{itemize}
\item $(\widetilde{\Omega},\widetilde{\cF},\widetilde{\F}_T,\widetilde{\Prob})$ contains $(\Omega, \cF, \F_T, \Prob)$.
\item $\widetilde{A} = A$ on $[0,\tau]$.
\item $\{ \tau < \widetilde{\tau} \leq T \} = \{ \tau < T \}$.
\item $\ud \langle \widetilde{A} \rangle_t = \Sigma \ud t$ on $t \in [\tau \wedge T, \widetilde{\tau} \wedge T)$.
\end{itemize}
\end{definition}

We will often refer to the main result below as the \textit{characterisation theorem}. 

\begin{theorem}
\label{theorem: characterisation of local martingales}
Let $\cM$ be a regular family of models for $A$ and denote $\cD_\cM \equiv \cR(X^B; \cM)$. Then, for any $F \in C(\cD_{\cM}) \cap C^2(\cD_{\cM} \setminus B)$, $F(X^B_\cdot)$ is a local martingale under all $M \in \cM$ if and only if $F \in \cS_X(\cD_{\cM} \setminus B)$ (defined in Definition \ref{definition: solution set to system of PDEs}).
\end{theorem}

We first make a few remarks regarding the theorem and then move on to its proof. For a regular class of models $\cM \subseteq \cM_{\ell}$ for $A$, it is clear that $\cM \subseteq \cM_\ell$ implies that $\cD_{\cM} \subseteq \cD_{\cM_\ell}$, where we recall that $\cD_{\cM}$ and $\cD_{\cM_\ell}$ are the ranges of $X^B$ over the classes of models $\cM$ and $\cM_\ell$ respectively. If we can find a regular class of models $\cM \subseteq \cM_\ell$ for $A$ such that $\cD_{\cM} = \cD_{\cM_\ell}$, we can conclude that  $F \in \cS_X(\cD_{\cM_\ell} \setminus B)$. This is summarised by the following corollary.

\begin{corollary}
If there is a regular family $\cM$ of models for $A$ such that $\cR(X^B; \cM) = \cR(X^B; \cM_{\ell}) \equiv \cD_{\cM_{\ell}}$, then for any $F \in C(\cD_{\cM_{\ell}}) \cap C^2(\cD_{\cM_{\ell}} \setminus B)$, $F(X^B_\cdot)$ is a local martingale under all local martingale models $M \in \cM_{\ell}$ iff $F \in \cS_X(\cD_{\cM_{\ell}} \setminus B)$.
\end{corollary}

\begin{lemma}
\label{lemma: regular class of models}
If $S$ is an underlying asset, then $\cM_{\ell}$ is a regular class of models for a market with a single risky asset ($A \equiv S$).
\end{lemma}

The previous lemma, which follows by directly checking the conditions in Definition \ref{definition: regular class of models}, implies that the replication strategies in Section \ref{section: examples} characterise the set of model-independent wealth processes of the form $F(X^B_\cdot)$ (for the various choices of $X \in \cX$ considered) for a market with $A \equiv S$.

\begin{proof}[Proof of Theorem \ref{theorem: characterisation of local martingales}]

Throughout the proof, we will denote $||f||_K$ to be the supremum norm of a function $f$ on a set $K$. The statement that $F(X^B_\cdot)$ is local martingale under all $M \in \cM$ follows from Proposition \ref{proposition: integral representation}. We hence turn to proving the converse statement.

Suppose that $F \in C(\cD_\cM) \cap C^2(\cD_\cM \setminus B)$ and that $F(X^B_\cdot)$ is a local martingale under all models in $\cM$. Consider any such model $\{ A; (\Omega, \cF, \F_T, \Prob) \} \in \cM$ and define
$$
Y_t[A] := F(X^B_t[A]) - \int_0^t \I_{\{ u < \tau^X_B \}} \cL^\alpha F(X^B_u) \ud A_u, \quad t \in [0,T].
$$
Note that $Y_t[A] = \int_0^t \I_{\{ u < \tau^X_B \}} \cL^\gamma F(X^B_u) \ud u + \sum_{i,j=1}^d \int_0^t \I_{\{ u < \tau^X_B \}} \cL^{\alpha, \beta}_{i,j} F(X^B_u) \ud \langle A^i, A^j \rangle_u$ by It\^o's formula. Hence, $Y_\cdot \equiv Y_\cdot[A]$ is a continuous local martingale of finite variation. By the Burkholder-Davis-Gundy inequality, it follows that $
\Prob( Y_t[A] = 0, \ t \in [0,T] ) = 1$, which implies that
\begin{equation}
\label{proof of characterisation theorem: eqn 1}
\E( Y_{\tau_2} - Y_{\tau_1} ) = 0
\end{equation}
for any $\tau_1, \tau_2 \in \cT(\F_T)$ such that $\tau_1 \leq \tau_2$.

We will now show that if $F \not\in \cS_X(\cD_\cM \setminus B)$, then there exists a model $M \in \cM$ and stopping times $\tau_1, \tau_2 \in \cT(\F_T)$ defined on it such that $\tau_1 \leq \tau_2$ and such that \eqref{proof of characterisation theorem: eqn 1} does not hold, and thereby obtain a contradiction to $F \not\in \cS_X(\cD_\cM \setminus B)$. The notation introduced in the following sections of the proof is temporary (it  only applies within each section).

\item (i) We first prove that $\cL^{\alpha,\beta}_{i,i} F = 0$ on $\cD_{\cM} \setminus B$ for all $i \in \{1,\ldots,d\}$. Suppose that for some $k \in \{1,\ldots,d\}$ and $x \in \cD_{\cM} \setminus B$,
$$
\cL^{\alpha,\beta}_{k,k} F(x) = z \neq 0.
$$
Without loss of generality, assume that $z > 0$ and $x \not\in \partial \cD_\cM$ (the proof when $z < 0$ is analogous, and if $x \in \partial \cD_\cM$, choose a point $x^\prime \in \cD_\cM \setminus \partial \cD_\cM$ close to $x$ such that $\cL^{\alpha,\beta}_{k,k} F(x) \neq 0$).

By continuity of $\cL^{\alpha,\beta}_{k,k} F$, $\exists \, \epsilon \in (0,d(x,B))$ such that
$$
\inf_{y \in B^c_x(\epsilon) \cap \cD_{\cM}} \cL^{\alpha,\beta}_{k,k} F(y) > z/2.
$$
Denote $K_1 := B^c_x(\epsilon/2) \cap \cD_{\cM}$ and $K_2 := B^c_x(\epsilon) \cap \cD_{\cM}$, and note that $\partial K_1 \cap \partial K_2 = \emptyset$ since $\epsilon < d(x,B)$. Also, define the bound $b$ and the matrix $\Sigma \in \mathbb{S}^d_+$ by
$$
b := || \cL^\gamma F ||_{K_2} \vee \max_{i,j \in \{1,\ldots,d\}} || \cL^{\alpha,\beta}_{i,j} F ||_{K_2}
$$
and
$$
\Sigma_{ij} :=
\begin{cases}
\frac{2b(d+1)}{z} \qquad& i = j = k, \\
1 & i = j \neq k, \\
0 & \textrm{otherwise}.
\end{cases}
$$
Since $x \in \cD$, $\exists \{ A ; (\Omega, \cF, \F_T, \Prob) \} \in \cM$ such that $\Prob(\tau^{X^B}_{K_1} < T) > 0$ (where $\tau^{X^B}_{K_1}$ stands for the first hitting time of $X^B_\cdot$ to the set $K_1$). Define the stopping time
$$
\tau_1 := \tau^{X^B}_{K_1} \wedge T \in \cT(\F_T).
$$
By $\Sigma \in \mathbb{S}^d_+$ and the regularity property of $\cM$, there exist a model $\{ \widetilde{A}; (\widetilde{\Omega}, \widetilde{\cF}, \widetilde{\F}_T, \widetilde{\Prob}) \} \in \cM$ and stopping time $\widetilde{\tau} \in \cT(\widetilde{\F}_T)$ such that the properties stated in Definition \ref{definition: regular class of models} hold. Define the stopping time
$$
\tau_2 := \widetilde{\tau} \wedge \inf\{ t \geq \tau_1 : X^B_t[\widetilde{A}] \in \partial K_2 \} \in \cT(\widetilde{\F}_T).
$$
Then,
\begin{align*}
Y_{\tau_2}[\widetilde{A}] - Y_{\tau_1}[\widetilde{A}] &= \int_{\tau_1}^{\tau_2} \cL^\gamma F(X^B_t[\widetilde{A}]) \ud t + \sum_{i,j=1}^d \int_{\tau_1}^{\tau_2} \cL^{\alpha,\beta}_{i,j} F(X^B_t[\widetilde{A}]) \ud \langle \widetilde{A}^i,\widetilde{A}^j \rangle_t \\
&= \int_{\tau_1}^{\tau_2} \left( \cL^\gamma F(X^B_t[\widetilde{A}]) + \sum_{i,j=1}^d \Sigma_{ij} \cL^{\alpha,\beta}_{i,j} F(X^B_t[\widetilde{A}]) \right) \ud t \\
&= \int_{\tau_1}^{\tau_2} \left( \frac{2b(d+1)}{z} \cL^{\alpha,\beta}_{k,k} F(X^B_t[\widetilde{A}]) + \cL^\gamma F(X^B_t[\widetilde{A}]) + \sum_{i \neq k} \cL^{\alpha,\beta}_{i,i} F(X^B_t[\widetilde{A}]) \right) \ud t \\
& \geq \int_{\tau_1}^{\tau_2} \left( \frac{2b(d+1)}{z} \inf_{y \in K_2} \cL^{\alpha,\beta}_{k,k} F(y) - || \cL^\gamma F ||_{K_2} - \sum_{i \neq k} || \cL^{\alpha,\beta}_{i,i} F||_{K_2} \right) \ud t \\
& \geq \int_{\tau_1}^{\tau_2} b \ud t
\end{align*}
Hence,
\begin{equation}
\label{proof of characterisation theorem: eqn 2}
Y_{\tau_2}[\widetilde{A}] - Y_{\tau_1}[\widetilde{A}] \geq b(\tau_2 - \tau_1)
\end{equation}
By construction of $\widetilde{\tau}$, it follows that $\{ \tau_1 < \widetilde{\tau} \} = \{ \tau_1 < T \}$, whereas $\partial K_1 \cap \partial K_2 = \emptyset$ implies that $\{ \tau_1 < \inf\{ t \geq \tau_1 : X^B_t[\widetilde{A}] \in \partial K_2 \} \} = \{\tau_1 < T \}$. Hence,
$$
\{ \tau_1 < \tau_2 \} = \{\tau_1 < \widetilde{\tau} \} \cap \{\tau_1 < \inf\{ t \geq \tau_1 : X^B_t[\widetilde{A}] \in \partial K_2 \} \} = \{ \tau_1 < T \},
$$
which implies that
$$
\widetilde{\Prob}(\tau_1 < \tau_2) = \widetilde{\Prob}(\tau_1 < T) = \Prob(\tau_1 < T) > 0.
$$
By \eqref{proof of characterisation theorem: eqn 2}, this gives $\E^{\widetilde{\Prob}}( Y_{\tau_2}[\widetilde{A}] - Y_{\tau_1}[\widetilde{A}] ) > 0$ (where $\E^{\widetilde{\Prob}}$ denotes the expectation under $\widetilde{\Prob}$), which contradicts \eqref{proof of characterisation theorem: eqn 1}. 

Hence, this proves that $\cL^{\alpha,\beta}_{i,i} F = 0$ on $\cD_{\cM} \setminus B$ for $i \in \{1,\ldots,d\}$.

\item (ii) We will now prove that $\cL^{\alpha,\beta}_{i,j} = 0$ on $\cD_{\cM} \setminus B$ for $i \neq j$. As in part (i), we argue by contradiction and provide a construction of a model and stopping times $\tau_1$ and $\tau_2$ such that \eqref{proof of characterisation theorem: eqn 1} does not hold. The main difference in this part of the proof is that the result for $i = j$ (proved in part (i)) is used in one of the steps.

Suppose that for some $k \neq l$ and $x \in \cD_{\cM} \setminus B$,
$$
\cL^{\alpha,\beta}_{k,l} F(x) = z \neq 0.
$$
Once again, assume without loss of generality that $z > 0$ and $x \not\in \partial \cD_\cM$. Let $\epsilon \in (0, d(x,B))$ be such that $\inf_{y \in B^c_x(\epsilon) \cap \cD_{\cM}} \cL^{\alpha,\beta}_{k,l} F(y) > z/2$, denote $K_1 := B^c_x(\epsilon/2) \cap \cD_{\cM}$, $K_2 := B^c_x(\epsilon) \cap \cD_{\cM}$ and define $b := || \cL^\gamma F ||_{K_2} \vee \max_{i \in \{1,\ldots,d\}} || \cL^{\alpha,\beta}_{i,j} ||_{K_2}$. Define the matrix $\Sigma \in \mathbb{S}^d_+$ by
$$
\Sigma_{ij} = 
\begin{cases}
\frac{5b}{z}, \qquad& i = j \in \{k,l\}, \\
\frac{4b}{z}, & (i,j) \in \{(k,l), (l,k)\}, \\
1, & i = j \not\in \{k,l\}, \\
0, & \textrm{otherwise}.
\end{cases}
$$
Since $x \in \cD_{\cM}$, $\exists \{ A; (\Omega, \cF, \F_T, \Prob) \} \in \cM$ such that $\Prob(\tau^{X^B}_{K_1}) > 0$. Define $\tau_1 := \tau^{X^B}_{K_1} \wedge T$, and note that, by $\Sigma \in \mathbb{S}^d_+$ and the regularity property of $\cM$, there exists a model $\{ \widetilde{A} ; (\widetilde{\Omega}, \widetilde{\cF}, \widetilde{\F}_T, \widetilde{P}) \} \in \cM$ and stopping time $\widetilde{\tau} \in \cT(\widetilde{\F}_T)$ such that properties stated in Definition \ref{definition: regular class of models} holds. Define the $\widetilde{\F}_T$-stopping time $\tau_2 := \widetilde{\tau} \wedge \inf\{ t \geq \tau_1 : X^B_t[\widetilde{A}] \in \partial K_2 \}$. Then,
\begin{align*}
Y_{\tau_2}[\widetilde{A}] - Y_{\tau_1}[\widetilde{A}] &= \int_{\tau_1}^{\tau_2} \cL^\gamma F(X^B_t[\widetilde{A}]) + \sum_{i,j=1}^d  \int_{\tau_1}^{\tau_2} \cL^{\alpha,\beta}_{i,j} F(X^B_t[\widetilde{A}]) \ud \langle \widetilde{A}^i, \widetilde{A}^j \rangle_t \\
&= \int_{\tau_1}^{\tau_2} \left( \cL^\gamma F(X^B_t[\widetilde{A}]) + \sum_{i,j=1}^d \Sigma_{ij} \cL^{\alpha,\beta}_{i,j} F(X^B_t[\widetilde{A}]) \right) \ud t \\
&= \int_{\tau_1}^{\tau_2} \left( 2 \left( \frac{4b}{z} \right) \cL^{\alpha,\beta}_{k,l} F(X^B_t[\widetilde{A}]) + \cL^\gamma F(X^B_t[\widetilde{A}]) \right) \ud t \\
&\geq \int_{\tau_1}^{\tau_2} \left( \frac{8b}{z} \inf_{y \in K_2} \cL^{\alpha,\beta}_{i,j} F(y) - || \cL^\gamma F||_{K_2} \right) \ud t \\
& \geq \int_{\tau_1}^{\tau_2} 3b \ud t.
\end{align*}
where we used $\cL^{\alpha,\beta}_{k,l} F = \cL^{\alpha,\beta}_{l,k} F$ and $\cL^{\alpha,\beta}_{i,i} F = 0$ on $\cD_\cM \setminus B$ for $i \in \{1,\ldots,d\}$ in obtaining the third line. Hence,
\begin{equation}
\label{proof of characterisation theorem: eqn 3}
Y_{\tau_2}[\widetilde{A}] - Y_{\tau_2}[\widetilde{A}] \geq 3b \, (\tau_2 - \tau_1).
\end{equation}
By analogous arguments as in part (i) of the proof, $\widetilde{\Prob}(\tau_1 < \tau_2) > 0$. Hence, with \eqref{proof of characterisation theorem: eqn 3}, this again implies that $\E^{\widetilde{\Prob}}(Y_{\tau_2}[\widetilde{A}] - Y_{\tau_1}[\widetilde{A}]) > 0$, which contradicts \eqref{proof of characterisation theorem: eqn 1}.

This proves that $\cL^{\alpha,\beta}_{i,j} F = 0$ on $\cD_{\cM} \setminus B$ for all $i \neq j$.

\item (iii) Parts (i) and (ii) proved that $\cL^{\alpha,\beta}_{i,j} F = 0$ on $\cD_{\cM} \setminus B$ for all $i,j \in \{1,\ldots,d\}$. Hence, \eqref{proof of characterisation theorem: eqn 1} becomes
$$
\E\left( \int_{\tau_1}^{\tau_2} \cL^\gamma F(X^B_t[\widetilde{A}]) \ud t \right) = 0.
$$
By a simpler version of the arguments in parts (i) and (ii), this implies that $\cL^\gamma F = 0$ on $\cD_{\cM} \setminus B$. This concludes the proof.
\end{proof}

% convex claim
\section{Underlying asset and convex claim}
\label{section: convex claim}

To provide a more concrete example to the results of the previous section, consider a market with an underlying asset $S$ and a traded European path-independent claim $C$ with maturity $T$ and payoff $g(S_T)$. Simple examples for $C$ would be a call or a put. We will assume that $g$ is a non-linear convex function such that $\int_\R g(s e^{y}) e^{-y^2/2} \ud y < \infty$ for all $s > 0$. The latter condition implies that the Black-Scholes price under $r=0$ and $\sigma=1$ is finite for all initial prices $S_0 = s$. 
Note that when we say that a convex function $g$ is non-linear on an interval $(s_1,s_2) \subset \R_+$, we mean that $g^\prime(s_1) < g^\prime(s_2)$, where $g^\prime$ denotes the left derivative of $g$. In the framework of our previous notation, we have $A = (S,C)$.

\begin{definition}
We define $V := C - g(S)$ to be the \textit{time value} of the claim $C$.
\end{definition}

\begin{notation}
\item (1) We say that $I$ is an $\F_T$-adapted open interval in $\R_+$ if $I = (a,b)$ for $a, b: \Omega \times [0,T] \to \R_+$ such that $a < b$ $\Prob$-a.e. and $a$, $b$ are $\F_T$-adapted. 

\item (2) For a probability space $(\Omega, \cF, \Prob)$, a $\sigma$-algebra $\cG \subseteq \cF$ and subsets $\xi_1, \xi_2 \in \cF$ of $\Omega$, $\Prob(\xi_1 \, | \, \cG) > 0$ on $\xi_2$ means that for any $\xi _3 \in \cG$ such that $\Prob(\xi_2 \cap \xi_3) > 0$, we have that $\Prob( \xi_1 \cap \xi_2 \cap \xi_3) > 0$. 

\item (3) A process $Z$ is a local martingale on $(\Omega, \cF, \F_T, \Prob)$ if there is a sequence of $\F_T$-stopping times $\tau_k \uparrow T$ such that $Z_{\cdot \wedge \tau_k}$ is a martingale with respect to $\F_T$.
\end{notation}

We now define a so-called \textit{full support} property of a stochastic process. This will be a central property in this section as well as in Section \ref{section: traded calls}. 

\begin{definition}
Let $S$ be an $\R_+$-valued stochastic process defined on a filtered probability space $(\Omega, \cF, \F_T, \Prob)$. We then say that $S$ has \textit{full support} or $S$ has the (FS) property if for any $\F_T$-adapted open interval $I \subset \R_+$ and for any $t \in [0,T)$, $\Prob(S_T \in I_t \, | \, \cF_t) > 0$. Similarly, for a subset $\xi \in \cF$ of $\Omega$, we say that $S$ has full support on $\xi$ if for any $\F_T$-adapted open interval $I \subset \R_+$ and for any $t \in [0,T)$, $\Prob(S_T \in I_t \, | \, \cF_t) > 0$ on $\xi$.
\end{definition}

The requirement that $\Prob(S_T \in I_t \, | \, \cF_t) > 0$ is an almost-sure statement on conditional probability distributions. It is equivalent to the statement $\Prob( \Prob(S_T \in I_t \, | \, \cF_t) = 0) = 0$. In particular, any process which is absorbed at a level which it can attain prior to time $T$ with non-zero probability does not have the full support property.

We can now define the family of models which we will prove to be regular for $A = (S,C)$.

\begin{definition}
\label{definition: regular class of models, convex claim}
Define $\cM \equiv \cM_g$ to be the family of models $\{ A ; (\Omega, \cF, \F_T, \Prob) \} \in \cM_\ell$ such that $C - \int_0^\cdot g^\prime(S_t) \ud S_t$ is a supermartingale and such that $S$ has full support.
\end{definition}

Note that $\cM \subseteq \cM_\ell$, hence $S$ and $C$ are local martingales (on the time interval $[0,T])$ such that $C_T = g(S_T)$ $\Prob$-a.e.. The class $\cM$ is general enough to work with models with bubbles in the price of the traded claim, notably such that $C_t > \E(g(S_T) | \cF_t)$. In particular, the class of models $\cM$ allows for strict local martingales for the traded options even if the underlying is a true martingale. This contrasts somewhat with the context considered in Cox and Hobson \cite{Cox-Hobson}, where the emphasis is on the asset price being a strict local martingale under the risk-neutral measure. The main result of this section is the following.

\begin{theorem}
\label{theorem: regular class of models, convex claim}
$\cM$ defined in Definition \ref{definition: regular class of models, convex claim} is a regular class of models for $A = (S,C)$. Furthermore, $\cM(A_0) \neq \emptyset$ iff $C_0 > g(S_0)$.
\end{theorem}

For $X \in \cX$ and a closed set $B \subset \R^n$, denote $\cD \equiv \cR(X^B; \cM)$ and consider a function $F \in C(\cD) \cap C^2(\cD \setminus B)$. Then, Theorem \ref{theorem: characterisation of local martingales} and Theorem \ref{theorem: regular class of models, convex claim} together imply that $F(X^B_\cdot)$ is a local martingale under all models in $\cM$ iff $F \in \cS_X(\cD \setminus B)$. We first present some results which will be helpful in proving Theorem \ref{theorem: regular class of models, convex claim}.

\begin{lemma}
\label{lemma: (FS) property with stopping times}
$S$ has the full support property iff for any $\F_T$-adapted stochastic open interval $I \subset \R_+$ and any stopping time $\tau \in \cT(\F_T)$, $\Prob(S_T \in I_\tau \, | \, \cF_\tau ) > 0$ on $\{ \tau < T \}$.
\end{lemma}

The proof of the above lemma is provided in the Appendix. Next, we derive a sufficient regularity condition on the volatility of a continuous local martingale $S$ in order for it to have full support. In particular, we show that processes which have a lower and upper bounded realised log-variance on $[\tau, T]$ for a stopping time $\tau$ are of full support on $\{ \tau < T\}$. Hence, coupling a stochastic process at a stopping time $\tau_{coupling}$ with a local volatility process with lower and upper bounded volatility will yield a process with full support on $\{ \tau_{coupling} < T \}$. The proof is again provided in the Appendix.

\begin{proposition}
\label{proposition: (FS) property with bounded realised variance}
Consider a filtered probability space $(\Omega, \cF, \F_T, \Prob)$ with a one-dimensional Brownian motion $W$, a left-continuous $\F_T$-adapted process $\sigma$ and a stopping time $\tau \in \cT(\F_T)$ defined on it. Suppose that $S$ is an $\F_T$-adapted process on $(\Omega, \cF, \F_T, \Prob)$ such that
$$
\ud S_t := \sigma_t S_t \ud W_t, \quad t \in [\tau \wedge T,T],
$$
Then, $S$ has full support on $\{ \tau < T\}$ if there are functions $l, u: [0,T] \to \R$ such that $0 < l(t) \leq u(t) < \infty$ for $t \in [0,T)$, $l(T) = u(T) = 0$ and
\begin{equation}
\label{lower bound for conditional realised variance}
\Prob \left( l(t \vee \tau) \leq \int_{t \vee \tau}^T \sigma^2_u \ud u \leq u(t \vee \tau) \, \bigg{|} \, \cF_{t \vee \tau} \right) > 0 \quad on \ \{ \tau < T \}
\end{equation}
for $t \geq 0$.
\end{proposition}

The following proposition is a key ingredient in the proof of Theorem \ref{theorem: regular class of models, convex claim}. 

\begin{proposition}
\label{proposition: open set for convex claim}
$\cM(A_0) \neq \emptyset$ iff $V_0 > 0$, and for any model $\{A;(\Omega,\cF,\F_T,\Prob) \} \in \cM$, $\Prob(V_t > 0, \ t \in [0,T) ) = 1$.
\end{proposition}

\begin{proof}

Suppose that $\{ A ; (\Omega, \cF, \F_T, \Prob) \} \in \cM$. The notational definitions in this proof are temporary. \\

\item (i) We begin by showing that for $a \in \R_+$, stopping times $\tau \in \cT(\F_T)$ and $\tau_a := \inf \{ t \geq \tau : S_t = a \}$, the full support property of $S$ implies that $\Prob( \tau_a < T \, | \, \tau < T) > 0$. Let $s_1 < a < s_2$. Then,
\begin{align*}
&\Prob( \tau_a < T \, | \, \tau < T) \\
&= \Prob( \tau_a < T, S_\tau = a \, | \, \tau < T) + \Prob( \tau_a < T, S_\tau < a \, | \, \tau < T) + \Prob( \tau_a < T, S_\tau > a \, | \, \tau < T) \\
&\geq \Prob( S_\tau = a \, | \, \tau < T) + \Prob( S_T > s_2, S_\tau < a \, | \, \tau < T) + \Prob( S_T < s_1, S_\tau > a \, | \, \tau < T) \\
&= \Prob( S_\tau = a \, | \, \tau < T) + \Prob( S_T > s_2 \, | \, S_\tau < a, \tau < T) \, \Prob(S_\tau < a \, | \, \tau < T) \\
&\hspace{150pt} + \Prob( S_T < s_1, \, | \, S_\tau > a, \tau < T) \, \Prob( S_\tau > a \, | \, \tau < T)
\end{align*}
Note that $\Prob( S_T > s_2 \, | \, S_\tau < a, \tau < T) > 0$ and $\Prob( S_T < s_1 \, | \, S_\tau > a, \tau < T) > 0$  by $\Prob(S_T > s_2 \, | \, \tau < T) > 0$ and $\Prob(S_T < s_1 \, | \, \tau < T) > 0$ and that  
$$
\Prob( S_\tau = a \, | \, \tau < T) + \Prob(S_\tau < a \, | \, \tau < T) + \Prob( S_\tau < a \, | \, \tau < T) = 1,
$$
We hence conclude that $\Prob( \tau_a < T \, | \, \tau < T) > 0$ as claimed. \\

\item (ii) Next, we show that for $a \in \R_+$ and $\tau \in \cT(\F_T)$, the full support property of $S$ implies that $\E(L^a_T[S] - L^a_{\tau \wedge T}[S]) > 0$ if $\Prob(\tau < T) > 0$.

Let $b \in \R_+ \setminus \{ a \}$ and define
\begin{align*}
\tau_a &:= \inf \{ t \geq \tau : S_t = a \}, \\
\tau_{a,b} &:= \inf \{ t \geq \tau_a : S_t = b \}.
\end{align*}
By monotonicity of $L^a_t[S]$ in $t$, Tanaka's formula and $\tau_a \wedge T \leq \tau_{a,b} \wedge T \leq T$, it follows that
\begin{align*}
\E \left( L^a_T[S] - L^a_{\tau \wedge T}[S] \right) &\geq \E \left( L^a_{\tau_{a,b} \wedge T}[S] - L^a_{\tau_a \wedge T}[S] \right) \\
&= \E \left( |S_{\tau_{a,b} \wedge T} - a| - |S_{\tau_a \wedge T} - a| \right) \\
&= \E \left( \left( |S_{\tau_{a,b} \wedge T} - a| - |S_{\tau_a \wedge T} - a| \right) \I_{\{ \tau_a < T \}} \right) \\
&= \E \left( |S_{\tau_{a,b} \wedge T} - a| \I_{\{ \tau_a < T \}} \right) \\
&\geq \E \left( |S_{\tau_{a,b} \wedge T} - a| \I_{\{ \tau_{a,b} < T \}} \right) \\
&\geq |b-a| \Prob( \tau_{a,b} < T )
\end{align*}
By $\{ \tau_{a,b} < T \} = \{ \tau_{a,b} < T, \tau_a < T \}$ and $\{ \tau_a < T \} = \{ \tau_a < T, \tau < T \}$,
\begin{align*}
\Prob(\tau_{a,b} < T) &= \Prob(\tau_{a,b} < T \, | \, \tau_a < T) \Prob( \tau_a < T) \\
&= \Prob(\tau_{a,b} < T \, | \, \tau_a < T) \Prob( \tau_a < T \, | \, \tau < T ) \Prob(\tau < T)
\end{align*}
Hence, by part (i) of the proof,
$$
\E \left( L^a_T[S] - L^a_{\tau \wedge T}[S] \right) \geq |b-a| \Prob(\tau_{a,b} < T \, | \, \tau_a < T) \, \Prob( \tau_a < T \, | \, \tau < T ) \, \Prob(\tau < T) > 0
$$
if $\Prob (\tau < T) > 0$. \\

\item (iii) Define $\tau^V_0 := \inf \{ t \geq 0 : V_t = 0 \}$. Since $V \geq 0$ and $V_T = 0$, it follows that $\tau^V_0 \leq T$. We will now show that $\Prob(\tau^V_0 < T) = 0$.

Temporarily define $Z := C - \int_0^\cdot g^\prime(S_t) \ud S_t$. The It\^o-Tanaka formula gives $V_t - V_0 = Z_t - Z_0 - \int_0^\infty L^a_t[S] \ud g^\prime(a)$, which along with $V_T = V_{\tau^V_0} = 0$ implies that 
$$
\int_0^\infty \left( L^a_T[S] - L^a_{\tau^V_0}[S] \right) \ud g^\prime(a) = Z_T - Z_{\tau^V_0}.
$$
Fubini's theorem and the supermartingale property of $Z$ (by the assumptions in Definition \ref{definition: regular class of models, convex claim}) then yield
\begin{align*}
\int_0^\infty \E \left( L^a_T[S] - L^a_{\tau^V_0}[S] \right) \ud g^\prime(a) &= \E \left( \int_0^\infty \left( L^a_T[S] - L^a_{\tau^V_0}[S] \right) \ud g^\prime(a) \right) \\
&= \E \left( Z_T - Z_{\tau^V_0} \right) \leq 0.
\end{align*}
By convexity and non-linearity of $g$ and by part (ii) of the proof, this implies that $\Prob( \tau^V_0 < T ) = 0$. \\

\item (iv) The conclusion in part (iii) proves that $V_t > 0$ for all $t \in [0,T)$ almost surely under all models in $\cM$. In particular, this implies that $\cM(A_0) = \emptyset$ if $V_0 \leq 0$. Hence, it remains to prove the converse claim that $\cM(A_0) \neq \emptyset$ if $V_0 > 0$. To do so, we will provide a construction of a model based on geometric Brownian motion and a strict local martingale.

Consider the canonical space $(\cC_\infty(\R^2),\cF^W,\F^W,\Prob^W)$ on which $(W^S,W^Y)$ is a two-dimensional Brownian motion. Define
$$
Y_t := 
\begin{cases}
\left(\frac{T-t}{T}\right)^{1/2} \exp\{W^Y_{\log(T/(T-t))}\}, \quad& t < T, \\
0, & t \geq T,
\end{cases}
$$
and denote $(\cC_T(\R^2), \cF, \F_T, \Prob)$ to be the canonical space of $\{ (W^S_t, Y_t) \}_{t \in [0,T]}$. $Y$ is a geometric Brownian motion (with zero mean) when stopped at a fixed sequence of times $t_k < T$ such that $t_k \uparrow T$. Hence, it is a local martingale on $(\cC_T(\R^2), \cF, \F_T, \Prob)$ starting at $Y_0 = 1$. Moreover, by $\lim_{t \to \infty} e^{W^Y_t - \frac{1}{2}t} = 0$ $\Prob^W$-a.e. (apply the law of iterated logarithm, Theorem 2.9.23 of Karatzas and Shreve), it follows that $\lim_{t \to T} Y_t = 0$ $\Prob^W$-a.e.. 

Temporarily define the Black-Scholes price of the claim on $g(S_T)$ at time zero under zero interest rate and volatility $\sigma > 0$ by 
$$
F_{BS}(s,t;\sigma) := \E \left( g \left( s \cE( \sigma W^S_t \right) \right).
$$
The following properties are well-known for Black-Scholes prices of convex claims, but we outline the arguments for completeness. By Fubini's theorem and the It\^o-Tanaka formula, $F_{BS}(S_0,T; \sigma) = g(S_0) + \int_0^t \E \left( L^a_T\left[S_0  e^{ \sigma W^1_T - \frac{1}{2} \sigma^2 T }\right] \right) \ud g^\prime(a)$. This implies that $F_{BS}(S_0,T; \sigma)$ is increasing in $\sigma$ and that $F_{BS}(S_0, T; \sigma) \downarrow g(S_0)$ as $\sigma \downarrow 0$ (by dominated convergence, using $\int_\R g(S_0 e^{y}) e^{-y^2/2} \ud y < \infty$ and $\lim_{\sigma \to 0} L^a_T\left[S_0  e^{ \sigma W^1_T - \frac{1}{2} \sigma^2 T }\right] = 0$). Hence,
$$
\sigma_0 := \sup\{ \sigma > 0 : F_{BS}(S_0, T; \sigma) - g(S_0) < V_0 \}
$$
is well-defined and $F_{BS}(S_0, T; \sigma_0) \leq C_0$. Then, for
\begin{align*}
S_t &:= S_0 e^{ \sigma_0 W^1_t - \frac{1}{2} \sigma_0^2 t }, \quad t \in [0,T], \\
C_t &:= F_{BS}(S_t, T-t; \sigma_0) + (C_0 - F_{BS}(S_0, T; \sigma_0)) Y_t, \quad t \in [0,T],
\end{align*}
we have that:

\item  (1) $S$ is a geometric Brownian motion, hence has full support by Proposition \ref{proposition: (FS) property with bounded realised variance}.

\item  (2) $C$ is a continuous local martingale on $(\cC_T(\R^3), \cF, \F_T, \Prob)$ such that $C_T = g(S_T)$.
\item (3) For $t \in [0,T]$,
\begin{align*}
&Z_t = C_t - \int_0^t g(S_u) \ud S_u \\
&= (C_0 - F_{BS}(S_0, T; \sigma_0)) Y_t + (F_{BS}(S_t, T - t; \sigma_0) - g(S_t)) + \int_0^\infty L^a_t [S] \ud g^\prime(a) + g(S_0) \\
&\geq g(S_0),
\end{align*}
which implies that $Z_\cdot$ is a lower bounded local martingale, thus a supermartingale. Hence, $\{ (S,C) ; (\cC_T(\R^2), \cF, \F_T, \Prob) \} \in \cM(A_0)$. This proves that $\cM(A_0) \neq \emptyset$ if $V_0 > 0$.

\end{proof}

\begin{proof}[Proof of Theorem \ref{theorem: regular class of models, convex claim}]

The claim that $\cM(A_0) \neq \emptyset$ iff $C_0 > g(S_0)$ was shown in Proposition \ref{proposition: open set for convex claim}. It remains to prove that $\cM$ is a regular class of models for $A = (S,C)$.

Consider a model $\{A ; (\Omega, \cF, \F_T, \Prob) \} \in \cM$, a stopping time $\tau \in \cT(\F_T)$ such that $\Prob( \tau < T) > 0$ and a matrix $\Sigma \in \mathbb{S}^2_+$. By similar arguments as in part (iv) of the proof of Proposition \ref{proposition: open set for convex claim}, consider a filtered probability space $(\Omega^\prime, \cF^\prime, \F^\prime_T, \Prob^\prime)$ on which are defined a $d$-dimensional Brownian motion $W$ and an independent continuous local martingale $Y$ on $(\cC_T(\R^3), \cF, \F_T, \Prob)$ starting at $Y_0 = 1$ and such that $\inf\{t \geq 0 : Y_t = 0\} = T$ $\Prob^\prime$-a.e.. Denote the Black-Scholes formula by
$$
F_{BS}(s,t;\sigma) := \E^{\Prob^\prime} \left( g\left( s \cE(\sigma W^1_t) \right) \right).
$$
For $c > g(s)$, define
$$
\sigma_t(s,c) := \sup\{ \sigma \geq 0 : F_{BS}(s,t; \sigma) \leq c \}.
$$
Then, $F_{BS}(s,t;\sigma_t(s,c)) \leq c$, with strict equality holding if and only if $\sigma_t(s,c) = \infty$. Define $(\widetilde{\Omega}, \widetilde{\cF},\widetilde{\F}_T,\widetilde{\Prob}) := (\Omega, \cF, \F_T, \Prob) \otimes (\Omega^\prime, \cF^\prime, \F^\prime, \Prob^\prime)$ and let $m \in \R^{2 \times 2}$ be such that $m^T m = \Sigma$. We can now define a model $\{ \widetilde{A}, (\widetilde{\Omega},\widetilde{\cF},\widetilde{\F}_T,\widetilde{\Prob}) \} \in \cM$ and a stopping time $\widetilde{\tau}$ satisfying the properties stated in Definition \ref{definition: regular class of models}. In particular, define
\begin{align*}
&\widetilde{A}_t := A_{t \wedge \tau} + m( W_t - W_{t \wedge \tau}), \quad t \in [\tau, \widetilde{\tau}), \\
&\widetilde{\tau} := \inf \left\{ t \geq \tau : \widetilde{C}_t - g(\widetilde{S}_t) \leq V_\tau/2 \right\} \wedge \frac{T + \tau}{2},
\end{align*}
and, using the shorthand notation $\widetilde{\sigma} \equiv \sigma_{\widetilde{\tau}}(\widetilde{S}_{\widetilde{\tau}},\widetilde{C}_{\widetilde{\tau}})$,
$$
\widetilde{S}_t := \widetilde{S}_{\widetilde{\tau}} \frac{\cE( \widetilde{\sigma} W^1_t)}{\cE( \widetilde{\sigma} W^1_{\widetilde{\tau}})}, \quad 
\widetilde{C}_t := BS(\widetilde{S}_t, t; \widetilde{\sigma}) + \frac{\widetilde{C}_{\widetilde{\tau}} - BS(\widetilde{S}_{\widetilde{\tau}}, \widetilde{\tau}; \widetilde{\sigma})}{Y_{\widetilde{\tau}}} Y_t, \quad t \in [\widetilde{\tau},T].
$$
By construction, $\widetilde{A} = A$ on $[0,\tau]$, $\widetilde{S}$ is a continuous martingale and $\widetilde{C}$ is a continuous local martingale such that $\widetilde{C}_T = g(\widetilde{S}_T)$. Denote $\widetilde{Z}_\cdot = \widetilde{C}_\cdot - \int_0^\cdot g^\prime(\widetilde{S}_t) \ud \widetilde{S}_t$. Note that $\widetilde{Z}$ is a local martingale. Furthermore,
\begin{align*}
\widetilde{Z}_t &= \widetilde{Z}_0 + \widetilde{C}_t - \widetilde{C}_0 - \left( g(\widetilde{S}_t) - g(\widetilde{S}_0) - \int_0^\infty L^a_t[\widetilde S] \ud g^\prime(a) \right)  \\
&= g(S_0) + \left( BS(\widetilde{S}_t, t; \widetilde{\sigma}) + \frac{\widetilde{C}_{\widetilde{\tau}} - BS(\widetilde{S}_{\widetilde{\tau}}, \widetilde{\tau}; \widetilde{\sigma})}{Y_{\widetilde{\tau}}} Y_t \right) - g(\widetilde{S}_t) + \int_0^\infty L^a_t[\widetilde S] \ud g^\prime(a) \\
& \geq g(S_0)
\end{align*}
by the convexity of $g$ (see the standard arguments outlined in part (iv) of proof of Proposition \ref{proposition: open set for convex claim}). This implies that $\widetilde{Z}$ is a lower bounded lower martingale, hence a supermartingale.

To prove that $\widetilde{S}$ is of full support, we will prove that it is of full support on sets $\xi_1, \xi_2 \in \widetilde{\cF}$ such that $\xi_1 \cup \xi_2 = \widetilde{\Omega}$. In particular, we will set $\xi_1 = \{ \widetilde{\tau} < T \}$ and $\xi_2  = \{ \widetilde{\tau} \geq T \}$. Then:
\item (1)  On $\{ \widetilde{\tau} < T \}$, $\ud \widetilde{S}_t = \widetilde{\sigma} \widetilde{S} \ud W^1_t$ for $t \in [\widetilde{\tau},T]$. Hence, $\widetilde{S}$ is of full support on $\{ \widetilde{\tau} < T \}$ by Proposition \ref{proposition: (FS) property with bounded realised variance}.
\item (2) On $\{ \widetilde{\tau} \geq T \}$, $\tau \geq T$ and hence $\widetilde{S} = S$ on $[0,T]$. Since $S$ is of full support, $\widetilde{S}$ is of full support on $\{ \widetilde{\tau} \geq T \}$. \\
This proves that $\widetilde{S}$ is of full support (on $\widetilde{\Omega} = \{ \widetilde{\tau} < T \} \cup \{ \widetilde{\tau} \geq T \}$) and that $\{ \widetilde{A}; (\widetilde{\Omega},\widetilde{\cF},\widetilde{\F}_T,\widetilde{\Prob}) \} \in \cM$. Furthermore, Proposition \ref{proposition: open set for convex claim} implies that $V_\tau > 0$ on $\{\tau < T\}$. Along with the continuity of paths of $\widetilde{A}$, this yields $\{\tau < \widetilde{\tau} \leq T \} = \{ \tau < T \}$. Also, note that $\ud \langle \widetilde{A} \rangle_t = \langle mW \rangle_t = m^T m \ud t$ for $t \in [\tau, \widetilde{\tau})$.

This concludes the proof that $\cM$ is a regular class of models for $A$.

\end{proof}

% examples
\section{Examples with $A = (S,C)$}
\label{section: examples with convex claim}

We now provide some examples of model-independent strategies for a market with an underlying asset and a traded claim with convex payoff function $g$. To simplify the presentation of the results, we will furthermore assume that $g \in C^2$. By known results \cite{Carr-Lee:volatility, Carr-Lee:semimartingales, Fukasawa, Hobson-Klimmek, Neuberger:log-contract}, the time value $V$ corresponds to the model-free implied weighted variance $\E \left( \int_t^T \frac{1}{2}g^{\prime\prime}(S_t) \ud \langle S \rangle_t \, \big{|} \, \cF_t \right)$ under any model where $S$ is a regular enough continuous local martingale. When $g(x) = - 2 \log(x)$, $V = C - g(S)$ corresponds to the theoretical value of the VIX index (without considering any rolling contract features and adjustments for jumps). Since the latter is a key index in industry, it is of interest to understand what payoffs contingent upon it and upon its second variations one may replicate model-independently (see Fukasawa \cite{Fukasawa} for a discussion of the role of the cross-variation between VIX and asset (log) returns). 

Within this section, we will denote $Q^s \equiv \int_0^\cdot w_s(S_t) \ud \langle S \rangle_t$, $Q^v \equiv \int_0^\cdot w_v(S_t) \ud \langle V \rangle_t$, $L \equiv \int_0^\cdot w_l(S_t) \ud \langle S, V \rangle_t$, $X \equiv (S, V, Q^s, Q^v, L)$, $X^s \equiv (S, V, Q^s)$, $X^v \equiv (S, V,Q^v)$ and $X^l \equiv (S, V, L)$, where $w_s, w_v, w_l > 0$ are continuous functions. It is easily checked that $X$, $X^s$, $X^v$ and $X^l$ are functionals in $\cX$. $\cM = \cM_g$ is defined as in Definition \ref{definition: regular class of models, convex claim}. Define $\cD = \cR(X; \cM) \setminus \{x_2 = 0\}$, $\cD^s = \cR(X^s; \cM) \setminus \{x_2 = 0\}$, $\cD^v = \cR(X^v; \cM) \setminus \{x_2 = 0\}$ and $\cD^l = \cR(X^l; \cM) \setminus \{x_2 = 0\}$. Note that $\cD \subset \R^5$ and that $\cD^s, \cD^v, \cD^l \subset \R^3$.

We start by simplifying the system of PDEs characterising $\cS_X(\cD)$.

\begin{lemma}
\label{lemma: PDEs for 5-d}
$F: \cD \to \R$ is in $\cS_X(\cD)$ iff $F \equiv F(s, v, q_s, q_v, l)$ is a solution to
$$
\begin{cases}
&w_s(s) \frac{\partial}{\partial q_s} F - \frac{1}{2} g^{\prime\prime}(s) \frac{\partial}{\partial v} F + \frac{1}{2} \frac{\partial^2}{\partial s^2} F = 0 \\
&w_v(s) \frac{\partial}{\partial q_v} F + \frac{1}{2} \frac{\partial^2}{\partial v^2} F = 0 \\
&w_l(s) \frac{\partial}{\partial l} F + \frac{\partial^2}{\partial s \partial v} F = 0
\end{cases}
$$
on $\cD$.
\end{lemma}

\begin{proof}

We start by noting that 
\begin{align*}
&\ud V_t = \ud C_t - g^\prime(S_t) \ud S_t - \frac{1}{2} g^{\prime\prime}(S_t) \ud \langle S \rangle_t, \\
&\ud \langle V \rangle_t = (g^\prime (S_t))^2 \ud \langle S \rangle_t + \ud \langle C \rangle_t - 2 g^\prime(S_t) \ud \langle S, C \rangle_t, \\
&\ud \langle S, V \rangle_t = - g^\prime(S_t) \ud \langle S \rangle_t + \ud \langle S, C \rangle_t
\end{align*}
For $F \equiv F(s,v,q_s, q_v, l) \in C^2(\R^5)$, It\^o's formula gives
\begin{align*}
\ud F(X_t) &= \frac{\partial}{\partial s} F(X_t) \ud S_t + \frac{\partial}{\partial v} F(X_t) \ud V_t \\
& \qquad + \frac{\partial}{\partial q_s} F(X_t) \ud Q^s_t + \frac{\partial}{\partial q_v} F(X_t) \ud Q^v_t + \frac{\partial}{\partial l} F(X_t) \ud L_t \\
& \qquad + \frac{1}{2} \frac{\partial^2}{\partial s^2} F(X_t) \ud \langle S \rangle_t + \frac{1}{2} \frac{\partial^2}{\partial v^2} F(X_t) \ud \langle V \rangle_t + \frac{\partial^2}{\partial s \partial v} F(X_t) \ud \langle S, V \rangle_t \\
&= \left( \frac{\partial}{\partial s} F(X_t) - g^\prime(S_t) \frac{\partial}{\partial v} F(X_t) \right) \ud S_t + \frac{\partial}{\partial v} F(X_t) \ud C_t \\
& \qquad + \left( w_s(S_t) \frac{\partial}{\partial q_s} F(X_t) + \frac{1}{2} \frac{\partial^2}{\partial s^2} F(X_t) - \frac{1}{2} g^{\prime\prime}(S_t) \frac{\partial}{\partial v} F(X_t) \right) \ud \langle S \rangle_t \\
& \qquad + \left( w_v(S_t) \frac{\partial}{\partial q_v} F(X_t) + \frac{1}{2} \frac{\partial^2}{\partial v^2}F(X_t) \right) \ud \langle V \rangle_t \\
&\qquad + \left( w_l(S_t) \frac{\partial}{\partial l} F(X_t) + \frac{\partial^2}{\partial s \partial v}F(X_t) \right) \ud \langle S, V \rangle_t
\end{align*}
Note that the integrand term corresponding to $\langle V \rangle$ must be equal to zero in order for the integral with respect to $\langle C \rangle$ to be equal to zero (this is because $\langle S, V \rangle$ does not contain any integrals with respect to $\langle C \rangle$). Then, in order for the integral with respect to $\langle S, C \rangle$ to be zero, the integral with respect to $\langle S, V \rangle$ must be zero. These two conditions along with the condition of the integral with respect to $\langle S \rangle$ being equal to zero yield
$$
\begin{cases}
&w_s(S_t) \frac{\partial}{\partial q_s} F(X_t) - \frac{1}{2} g^{\prime\prime}(S_t) \frac{\partial}{\partial v} F(X_t) + \frac{1}{2} \frac{\partial^2}{\partial s^2} F(X_t) = 0 \\
&w_v(S_t) \frac{\partial}{\partial q_v} F(X_t) + \frac{1}{2} \frac{\partial^2}{\partial v^2} F(X_t) = 0 \\
&w_l(S_t) \frac{\partial}{\partial l} F(X_t) + \frac{\partial^2}{\partial s \partial v} F(X_t) = 0.
\end{cases}
$$
Hence, the characterisation theorem implies that $F \equiv F(s, v, q_s, q_v, l)$ is in $\cS_X(\cD)$ iff
$$
\begin{cases} 
&w_s(s) \frac{\partial}{\partial q_s} F - \frac{1}{2} g^{\prime\prime}(s) \frac{\partial}{\partial v} F + \frac{1}{2} \frac{\partial^2}{\partial s^2} F = 0 \\
&w_v(s) \frac{\partial}{\partial q_v} F + \frac{1}{2} \frac{\partial^2}{\partial v^2} F = 0 \\
&w_l(s) \frac{\partial}{\partial l} F + \frac{\partial^2}{\partial s \partial v} F = 0,
\end{cases}
$$
on $\cD$, which concludes the proof. 

\end{proof}

By imposing further conditions on $F \in \cS_X(\cD)$, we obtain the systems of PDEs corresponding to $X^s$, $X^v$ and $X^l$ which can be solved in closed form.

\begin{proposition}
\label{proposition: MFIV}
For $F \in C^2(\cD^s) \cap C(\overline{\cD^s})$, $F(S_\cdot, V_\cdot, Q^s_\cdot)$ is a local martingale under all models in $\cM$ iff
\begin{equation}
\label{MFIV}
F(s, v, q_s) = c_1(v + g(s)) + c_2 \left( \frac{g^{\prime\prime}(s)}{w_s(s)} q_s - g(s) \right) + F_h(s,q_s)
\end{equation}
for a constant $c \in \R$ and a solution $F_h: \R_+^2 \to \R$ to the heat equation $
2 w_s(s) \frac{\partial} {\partial q_s} F_h(s, q_s) + \frac{\partial}{\partial s^2} F_h(s, q_s) = 0$. 
\end{proposition}

\begin{proof}

Since 
$$
\cS_{X^s}(\cD^s) = \left\{ \Ftilde \in \cS_X(\cD) : \frac{\partial}{\partial q_v} \Ftilde(s, v, q_s, q_v, l) = \frac{\partial}{\partial l} \Ftilde(s, v, q_s, q_v, l) = 0 \right\},
$$
Lemma \ref{lemma: PDEs for 5-d} implies that $F \in \cS_{X^s}(\cD^s)$ iff
\begin{align}
&w_s(s) \frac{\partial}{\partial q_s} F(s, v, q_s) - \frac{1}{2} g^{\prime\prime}(s) \frac{\partial}{\partial v} F(s, v, q_s) + \frac{1}{2} \frac{\partial^2}{\partial s^2} F(s, v, q_s) = 0, \label{PDE 1 for X^s} \\
&\frac{\partial^2}{\partial v^2} F(s, v, q_s) = 0, \label{PDE 2 for X^s} \\
&\frac{\partial^2}{\partial s \partial v} F(s, v, q_s) = 0 \label{PDE 3 for X^s}
\end{align}
on $\cD^s$. By \eqref{PDE 2 for X^s} and \eqref{PDE 3 for X^s}, $F$ must be of the form
$$
F(s, v, q_s) = v f_3(q_s) + f_{13}(s, q_s)
$$
for functions $f_3$ and $f_{13}$. Applying \eqref{PDE 1 for X^s}, this yields
\begin{equation}
\label{X^s proof intermediate equation}
w_s(s) v f_3^\prime(q_s) + w_s(s) \frac{\partial}{\partial q_s} f_{13}(s, q_s) - \frac{1}{2} g^{\prime\prime}(s) f_3(q_s) + \frac{1}{2} \frac{\partial^2}{\partial s^2} f_{13}(s, q_s) = 0.
\end{equation}
Differentiating the above equation with respect to $v$ gives $f_3^\prime (q_s) = 0$. Hence $F(s, v, q_s) = c_1 v + f_{13}(s, q_s)$ for a constant $c_1 \in \R$. Hence, \eqref{PDE 1 for X^s} gives
$$
2 w_s(s) \frac{\partial}{\partial q_s} f_{13}(s, q_s) + \frac{\partial^2}{\partial s^2} f_{13} (s, q_s) =  c_1 g^{\prime\prime}(s).
$$
It follows that $f_{13}(s, q_s) = F_h(s, q_s) + c_2 \frac{g^{\prime\prime}(s)}{w_s(s)} q_s + (c_1 - c_2) g(s)$ for a solution $F_h$ to the heat equation $2 w_s(s) \frac{\partial}{\partial q_s} F_h(s, q_s) + \frac{\partial^2}{\partial s^2} F_h(s, q_s) = 0$. This implies that $F \in \cS_{X^s}(\cD^s)$ iff $F(s, v, q_s) = c_1(v + g(s)) + c_2 \left( \frac{g^{\prime\prime}(s)}{w_s(s)} q_s - g(s) \right) + F_h(s,q_s)$, which along with Theorem \ref{theorem: characterisation of local martingales} and Theorem \ref{theorem: regular class of models, convex claim} concludes the proof.

\end{proof}

The previous proposition implies that the only functions of (weighted) realised variance replicable model-independently via strategies with wealth process of the form $F(X^s)$ are linear combinations of strategies in Subsection \ref{subsection: price and quadratic variation} and of the well-known strategy for replicating VIX. The following proposition shows that there are no functions of $Q^v_T$ replicable model-independently via strategies with wealth process of the form $F(X^v)$.
\begin{proposition}
\label{proposition: MFVV}
For $F \in C^2(\cD^v) \cap C(\overline{\cD^v})$, $F(S_\cdot, V_\cdot, Q^v_\cdot)$ is a local martingale under all models in $\cM$ iff
\begin{equation}
\label{MFVV}
F(s, v, q_v) = c_1 (g(s) + v) + c_2 s + c_3
\end{equation}
for constants $c_1, c_2, c_3 \in \R$.
\end{proposition}

\begin{proof}

Since 
$$
\cS_{X^v}(\cD^v) = \left\{ \Ftilde \in \cS_X(\cD) : \frac{\partial}{\partial q_s} \Ftilde(s, v, q_s, q_v, l) = \frac{\partial}{\partial l} \Ftilde(s, v, q_s, q_v, l) = 0 \right\},
$$
Lemma \ref{lemma: PDEs for 5-d} implies that $F \in \cS_{X^v}(\cD^v)$ iff
\begin{align} 
&g^{\prime\prime}(s) \frac{\partial}{\partial v} F(s, v, q_v) - \frac{\partial^2}{\partial s^2} F(s, v, q_v) = 0, \label{PDE 1 for X^v} \\
&w_v(s) \frac{\partial}{\partial q_v} F(s, v, q_v) + \frac{1}{2} \frac{\partial^2}{\partial v^2} F(s, v, q_v) = 0, \label{PDE 2 for X^v} \\
&\frac{\partial^2}{\partial s \partial v} F(s, v, q_v) = 0 \label{PDE 3 for X^v}
\end{align}
on $\cD^v$. By \eqref{PDE 3 for X^v}, $F$ must be of the form
$$
F(s, v, q_v) = f_{13}(s, q_v) + f_{23}(v, q_v)
$$
for functions $f_{13}$ and $f_{23}$. Then, \eqref{PDE 1 for X^v} gives
\begin{align*}
&g^{\prime\prime}(s) \frac{\partial}{\partial v} f_{23}(v, q_v) - \frac{\partial^2}{\partial s^2} f_{13}(s, q_v) = 0\\
\Rightarrow \quad &\frac{\partial^2}{\partial v^2} F(s, v, q_v) = \frac{\partial^2}{\partial v^2} f_{23}(v, q_v) = 0.
\end{align*}
By \eqref{PDE 2 for X^v}, $\frac{\partial}{\partial q_v} F(s, v, q_v) = -\frac{1}{w_v(s)} \frac{\partial^2}{\partial v^2} F(s, v, q_v) = 0$, hence $F$ must be of the form $F(s, v, q_v) = f_1(s) + c_1 v$
for a constant $c_1 \in \R$ and a function $f_1$. Finally, by \eqref{PDE 1 for X^v},
$$
\frac{\partial^2}{\partial s^2} f_1(s) = c_1 g^{\prime\prime}(s) \quad \Rightarrow \quad f_1(s) = c_1 g(s) + c_2 s + c_3
$$
for constants $c_2, c_3 \in \R$. Hence, $F \in \cS_{X^v}(\cD^v)$ iff $F(s, v, q_v) = c_1 (g(s) + v) + c_2 s + c_3$. Along with Theorem \ref{theorem: characterisation of local martingales} and Theorem \ref{theorem: regular class of models, convex claim}, this concludes the proof.

\end{proof}

For a final example in this section, we start by defining the function
\begin{equation}
\label{MFIL vanilla function}
G(x) := \int_0^{x} \int_0^z z g^{\prime\prime}(z) \ud z \ud u, \quad x > 0.
\end{equation}
We now show that one can replicate $\frac{1}{w_l(S_T)} L_T$ by a combination of a portfolio with wealth process $F(X_\cdot)$ and a static position in co-maturing claims on $G(S_T)$ if $\frac{1}{w_l(s)}$ is affine (this model-independent replication strategy was shown by Fukasawa \cite{Fukasawa} for $w_l(s) = \frac{1}{s}$), and that there are no non-linear functions of $L_T$ one can replicate by a portfolio of the form $F(X^l_\cdot)$. The resulting strategy trades dynamically in both the underlying asset and the traded claim $C$ (with payoff $g(S_T)$).

\begin{proposition}
\label{proposition: MFIL}
For $F \in C^2(\cD^l) \cap C(\overline{\cD^l})$, $F(S_\cdot, V_\cdot, L_\cdot)$ is a local martingale under all models in $\cM$ iff
\begin{equation}
\label{MFIL}
F(s, v, l) = c_1 \left( s v + \frac{1}{2} G(s) - \frac{l}{w_l(s)} \right) + c_2 (v + g(s)) + c_3 s + c_4
\end{equation}
for constants $c_1, c_2, c_3, c_4 \in \R$, where $c_1 = 0$ if $\frac{1}{w_l(s)}$ is not affine.
\end{proposition}

\begin{proof}

Since 
$$
\cS_{X^l}(\cD^l) = \left\{ \Ftilde \in \cS_X(\cD) : \frac{\partial}{\partial q_s} \Ftilde(s, v, q_s, q_v, l) = \frac{\partial}{\partial q_v} \Ftilde(s, v, q_s, q_v, l) = 0 \right\},
$$
Lemma \ref{lemma: PDEs for 5-d} implies that $F \in \cS_{X^l}(\cD^l)$ iff
\begin{align} 
&g^{\prime\prime}(s) \frac{\partial}{\partial v} F(s, v, l) - \frac{\partial^2}{\partial s^2} F(s, v, l) = 0, \label{PDE 1 for X^l} \\
&\frac{\partial^2}{\partial v^2} F(s, v, l) = 0, \label{PDE 2 for X^l} \\
&w_l(s) \frac{\partial}{\partial l} F(s, v, l) + \frac{\partial^2}{\partial s \partial v} F(s, v, l) = 0 \label{PDE 3 for X^l}
\end{align}
on $\cD^l$. By \eqref{PDE 2 for X^l}, $F$ must be of the the form $F(s, v, l) = v f^{(a)}(s, l) + f^{(b)}(s, l)$ for functions $f^{(a)}_{13}$ and $f^{(b)}_{13}$. By \eqref{PDE 3 for X^l},
$$
w_l(s) v \frac{\partial}{\partial l} f^{(a)}(s, l) + w_l(s) \frac{\partial}{\partial l} f^{(b)}(s, l) + \frac{\partial}{\partial s} f^{(a)}(s, l) = 0.
$$
Differentiating the above identities with respect to $v$ gives $\frac{\partial}{\partial l} f^{(a)}(s, l) = 0$. Hence, $F$ must be of the form $F(s, v, l) = v f_1(s) + f_{13}(s, l)$ for functions $f_1 \equiv f^{(a)}$ and $f_{13} \equiv f^{(b)}$. By \eqref{PDE 1 for X^l},
$$
g^{\prime\prime}(s) f_1(s) = v f_1^{\prime\prime}(s) + \frac{\partial^2}{\partial s^2} f_{13}(s, l).
$$
Differentiating the above with respect to $v$ gives
$$
f_1^{\prime\prime}(s) = 0 \quad \Rightarrow \quad f_1(s) = c_1 s + c_2
$$
for constants $c_1, c_2 \in \R$. Hence, $F(s, v, l) = v (c_1 s + c_2) + f_{13}(s, l)$. By \eqref{PDE 1 for X^l} and \eqref{PDE 3 for X^l}, it then follows that
$$
\begin{cases}
\frac{\partial^2}{\partial s^2} f_{13}(s, l) = g^{\prime\prime}(s) (c_1 s + c_2), \\
w_l(s) \frac{\partial}{\partial l} f_{13}(s, l) + c_1 = 0.
\end{cases}
$$
This implies that
$$
f_{13}(s, l) = -\frac{1}{w_l(s)} c_1 l + c_1 G(s) + c_2 g(s) + c_3 s + c_4
$$
for constants $c_2, c_3, c_4 \in \R$, and that $c_1 = 0$ if $\frac{1}{w_l(s)}$ is not affine. Hence,
$$
F(s, v, l) = c_1 \left( s v + G(s) - \frac{l}{w_l(s)} \right) + c_2(v + g(s)) + c_3 s + c_4,
$$
with $c_1 = 0$ if $\frac{1}{w(s)}$ is not affine. Along with Theorem \ref{theorem: characterisation of local martingales} and Theorem \ref{theorem: regular class of models, convex claim}, this concludes the proof.

\end{proof}

\section{Underlying asset and traded calls}
\label{section: traded calls}

As a final example, we consider a market with an underlying asset $S$ and a set of co-maturing traded calls $C^i$ written on $S$. Since the case of a single traded call option of strike $k > 0$ is covered by Section \ref{section: convex claim} with $g(x) = (x-k)^+$, we assume that $d \geq 3$. In particular, we consider the market $A = (S, C^2, \ldots, C^d)$, where the payoff requirements are $C^i_T = (S_T - k_i)^+$ for strikes $k_i$ which are in increasing order ($0 = k_1 < k_2 < \ldots < k_d < \infty$). Note that $C^1 \equiv S$ by convention. As in Section \ref{section: convex claim}, we will define a class of models $\cM$, elaborate on some properties of models in this class and show that $\cM$ is regular. By Theorem \ref{theorem: characterisation of local martingales}, the latter will imply that the characterisation result applies with this choice of class of models. Since many of the proofs of the results in this section are similar to proofs of analogous results in Section \ref{section: convex claim}, we will provide less details in some of them. We begin by defining the class of models $\cM$.

\begin{definition}
\label{definition: regular class of models, call options}
Define $\cM \equiv \cM_{calls}$ to be the set of models $\{ A; (\Omega, \cF, \F_T, \Prob) \} \in \cM_m$ such that $S$ has full support.
\end{definition}

Note that $\cM \subseteq \cM_m$, hence $S$ and $C$ are continuous martingales such that for all $i$, $C^i_T = (S^i_T - k_i)^+$ $\Prob$-a.e.. Hence the class $\cM = \cM_{calls}$ in Definition \ref{definition: regular class of models, call options} only considers models under which traded option prices are equal to conditional expectations of their payoffs. In order to construct models in $\cM$, it will therefore be sufficient to construct continuous martingales $S$ having full support. In fact, the main proofs below will construct local volatility models with volatility bounded within a closed interval in $(0,\infty)$ (such models have full support by Proposition \ref{proposition: (FS) property with bounded realised variance}).

\begin{definition}
\label{definition: time value and discretised slope}
For a martingale $S$, a stopping time $\tau \in \cT(\F_T)$ and $k: \Omega \times [0,T] \to \R_+$ which is $\F_T$-adapted, define $c_\tau(k) := \E((S_T - k_\tau)^+ \, | \, \cF_{\tau \wedge T})$. Denote $c^\prime_\tau(\cdot)$ to be the left-derivative of $c_\tau(\cdot)$.\end{definition}

For constants $k > 0$ and $t \in [0,T]$, $c_t(k)$ is equal to the call price with strike $k$ at time $t$.

\begin{definition}
\label{definition: time value and discretised slope}
For $i \in \{ 2, \ldots,d \}$, define $V^i := C^i - (S - k_i)^+$ and  $D^i := \frac{C^{i} - C^{i-1}}{k_i - k_{i-1}}$.  For $i \in \{ 3, \ldots,d \}$, define $\Delta^i := D^i - D^{i-1}$. Also, define $V^{min} := \min_{i \in \{2,\ldots,d\}} V^i$ and $\Delta^{min} := \min_{i \in \{3,\ldots,d\}} \Delta^i$. 
\end{definition}

The $V^i$ correspond to the time values of the call options $C^i$, whereas the $D^i$ correspond to the left-hand slope of the linear interpolation of the traded call prices. Note that $\Delta^{min}$ is positive if and only if the linear interpolation of call prices has a kink at each strike $k_i$ ($i \in \{2, \ldots, d-1 \}$). We now show that the (FS) property of a local martingale $S$ is equivalent to strict convexity of the call price function it generates.

\begin{lemma}
\label{lemma: equivalence between (FS) and strict convexity}
The following are equivalent statements of the (FS) property:
\begin{enumerate}
\item For any $\F_T$-adapted open interval $I$ and any $t \in [0,T)$, $
\Prob( S_T \in I_t \, | \, \cF_t) > 0$.
\item For any $\F_T$-adapted open interval $I$ and any $\tau \in \cT(\F_T)$, $\Prob( S_T \in I_{\tau \wedge T} \, | \, \cF_\tau) > 0$ on $\{ \tau < T \}$.
\item For any $\F_T$-adapted $0 < a < b < \infty$ and for any $t \in [0,T)$, $
c^\prime_t(a) < c^\prime_t(b)$.
\item For any $\F_T$-adapted $0 < a < b < \infty$ and any $\tau \in \cT(\F_T)$, $c^\prime_\tau(a) < c^\prime_\tau(b)$ on $\{\tau < T\}$.
\end{enumerate}
\end{lemma}

\begin{proof}

$(1) \Leftrightarrow (2)$ was proved by Lemma \ref{lemma: (FS) property with stopping times}. $(4) \Rightarrow (3)$ is trivial. $(3) \Rightarrow (1)$ follows from noting that for an $\F_T$-adapted open interval $I = (a,b)$ and for $t \in [0,T)$,
$$
\Prob( S_T \in I_t \, | \, \cF_t) \geq \Prob( S_T \in [(a_t + b_t)/2, b_t) \, | \, \cF_t ) = c^\prime_t(b) - c^\prime_t( (a+b)/2 ) > 0.
$$
Similarly, $(2) \Rightarrow (4)$ follows from noting that for $\tau \in \cT(\F_T)$ such that $\Prob(\tau < T) > 0$ and for $\F_T$-adapted $0 < a < b < \infty$,
$$
c^\prime_\tau(b) - c^\prime_\tau(a) = \Prob( S_T \in [a_\tau, b_\tau) \, | \, \cF_\tau) \geq \Prob( S_T \in (a_\tau, b_\tau) \, | \, \cF_\tau) > 0 \quad \textrm{on} \ \{\tau < T\}.
$$
\end{proof}

The following theorem is the main result of this section. It implies that the replication results in Section \ref{section: replication theory} characterise the full set of model-independent identities and replication results for the set of models $\cM$ defined in Definition \ref{definition: regular class of models, call options}.

\begin{theorem}
\label{theorem: regular class of models, call options}
$\cM$ defined in Definition \ref{definition: regular class of models, call options} is a regular class of models for $A = (S,C^2,\ldots,C^d)$. Furthermore, $\cM(A_0) \neq \emptyset$ iff $V^{min}_0 > 0$, $\Delta^{min}_0 > 0$, $D^2_0 > -1$ and $D^d_0 < 0$.
\end{theorem}

Theorem \ref{theorem: regular class of models, call options} is the analogue of Theorem \ref{theorem: regular class of models, convex claim} for the market considered in this section. Recall that the proof of Theorem \ref{theorem: regular class of models, convex claim} relied on 
Proposition \ref{proposition: open set for convex claim}. Similarly, the proof of Theorem \ref{theorem: regular class of models, call options} relies on the following proposition.

\begin{proposition}
\label{proposition: open set for call options}
$\cM(A_0) \neq \emptyset$ iff $V^{min}_0 > 0$, $\Delta^{min}_0 > 0$, $D^2_0 > -1$ and $D^d_0 < 0$, and for any $\{A;(\Omega,\cF,\F_T,\Prob) \} \in \cM$,
$$
\Prob(V^{min}_t > 0, \Delta^{min}_t > 0, D^2_t > -1, D^d_t < 0, \ \forall t \in [0,T) ) = 1.
$$
\end{proposition}

\begin{proof}

Temporarily define
\begin{align*}
Z &:= V^{min} \wedge \Delta^{min} \wedge (D^2 + 1) \wedge (-D^d), \\
\tau_0 &:= \inf\{ t \geq 0: Z_t \leq 0 \}.
\end{align*}
In part (i) of the proof, we will first prove that $\tau_0 = T$ $\Prob$-a.e., which will imply that
$$
\Prob(V^{min}_t > 0, \Delta^{min}_t > 0, D^2_t > -1, D^d_t < 0, \ \forall t \in [0,T) ) = 1,
$$
and that $\cM(A_0) = \emptyset$ if $A_0 = (S_0, C^1_0, \ldots, C^d_0)$ is such that any of the conditions $V^{min}_0 > 0$, $\Delta^{min}_0 > 0$, $D^2_0 > -1$ or $D^d_0 < 0$ does not hold. In part (ii), we will prove that $\cM(A_0) \neq \emptyset$ if $V^{min}_0 > 0$, $\Delta^{min}_0 > 0$, $D^2_0 > -1$ and $D^d_0 < 0$. \\

\item (i) Temporarily introduce the notation
\begin{align*}
\tau^V &:= \inf\{ t \geq 0: V^{min} \leq 0 \}, \\
\tau^\Delta &:= \inf\{ t \geq 0: \Delta^{min}_t \leq 0 \}, \\
\tau^D_l &:= \inf\{ t \geq 0: D^2_t \leq -1 \}, \\
\tau^D_r &:= \inf\{ t \geq 0: D^d_t \geq 0 \},
\end{align*}
and note that $\tau_0 = \tau^V \wedge \tau^\Delta \wedge \tau^D_l \wedge \tau^D_r$.

For any $i \in \{2,\ldots,d\}$ and any $\tau \in \cT(\F_T)$,
\begin{align*}
\E( V^i_{\tau \wedge T} ) = \E \left( (S_T - k_i)^+ - (S_{\tau \wedge T} - k_i)^+ \right) = \frac{1}{2} \E \left( L^{k_i}_T[S] - L^{k_i}_{\tau \wedge T}[S]  \right).
\end{align*}
Thus, part (ii) of the proof of Proposition \ref{proposition: open set for convex claim} implies that $\E( V^i_{\tau \wedge T} ) > 0$ if $\Prob(\tau < T) > 0$. It then follows that $\Prob( \inf\{ t \geq 0: V^i_t \leq 0 \} < T) = 0$ for $i \in \{2,\ldots,d\}$. Hence, $\tau^V = T$ $\Prob$-a.e.. 

To conclude the proof of $\tau_0 = T$ $\Prob$-a.e., we will show that $\tau^\Delta \vee \tau^D_l \vee \tau^D_r \geq T$ $\Prob-a.e.$. For any $\tau \in \cT(\F_T)$ and any $k > 0$, we have that $c^\prime_\tau(k) = - \Prob( S_T \geq k \, | \, \cF_{\tau \wedge T}) \in [-1,0]$. Furthermore, by the (FS) property of $S$ and Lemma \ref{lemma: equivalence between (FS) and strict convexity}, $c_\tau(\cdot)$ is a strictly convex function on $\{ \tau < T \}$ for any $\tau \in \cT(\F_T)$. This implies that $D^2_{\tau^D_l} > \lim_{k \downarrow 0} c^\prime_{\tau^D_l}(k) \geq -1$ on $\{\tau^D_l < T\}$, which contradicts $D^2_{\tau^D_l} = -1$. Hence, $\Prob( \tau^D_l \geq T) = 1$. By similar arguments by contradiction, $c^\prime_{\tau^D_r}(k_d) \geq 0$ and strict convexity of $c^\prime_{\tau^D_r}(\cdot)$ on $\{ \tau^D_r < T \}$ imply that $\Prob( \tau^D_r \geq T ) = 1$, and strict convexity of $c^\prime_{\tau^\Delta}(\cdot)$ on $\{ \tau^\Delta < T \}$ implies that $\Prob( \tau^\Delta \geq T ) = 1$. This concludes part (i) of the proof. \\

\item (ii) Suppose that $A_0$ is such that $V^{min}_0 > 0$, $\Delta^{min}_0 > 0$, $D^2_0 > -1$ and $D^d_0 < 0$. To show that $\cM(V^0) \neq \emptyset$, we will construct a local volatility process 
$$
S_\cdot = S_0 + \int_0^\cdot \sigma(S_t, t) S_t \ud W_t
$$
on the canonical space $(\cC(\R), \cF^W, \F^W_T, \Prob^W)$ for a Brownian motion $W$ such that:
\item (1) The corresponding call price function $c_0(k)$ passes through all of the points $(k_i, C^i_0)$.
\item (2) $\sigma(x) \in [l,u]$, $x > 0$ for $0 < l < du< \infty$. \\
Proposition \ref{proposition: (FS) property with bounded realised variance} will then imply that $S$ has full support. Defining the $C^i$ to be the conditional expectations of the call payoffs under $S$, this will yield a model in $\cM$.

Define $a_1 = (S_0 - C^2_0)/2$ and $a_2 := k_d - \frac{C^d_0}{D^d_0}$. By $D^2_0 \in (-1, 0)$ it follows that $a_1 \in (0, k_2 \wedge S_0)$ and that $k_2 - a_1 < k_2/2$. By $C^d_0 > V^{min}_0 > 0$ and $D^d_0 < 0$, it follows that $a_2$ is well-defined and $a_2 > k_d$. In particular, $a_2$ is the point on the horizontal (strike) axis where the latter is intersected by the line connecting $(k_{d-1},C^{d-1}_0)$ and $(k_d,C^d_0)$. By $D^2_0 > -1$, $\Delta^{min}_0 > 0$, it also follows that $a_2 > S_0$.

Denote $c^{BS}(k, t; \sigma)$ to be the call price function generated by the geometric Brownian motion with zero drift and volatility $\sigma$, notably
$$
c^{BS}(k, t; \sigma) := \E^{\Prob^W} ((S_0 \cE(\sigma W_t) - k)^+).
$$
Define 
\begin{align*}
\sigma_1 &:= \sup \left\{ \sigma > 0: c^{BS}(a_2, T; \sigma) - \frac{\partial}{\partial k} c^{BS}(a_2, T; \sigma) (k_d-a_2) \leq \frac{1}{2} C^d_0 \right\}, \\
\sigma_2 &:= \sup \left\{ \sigma > 0: \frac{\partial}{\partial k}c^{BS}(a_1, T; \sigma) \leq \frac{C^2_0 - c^{BS}(a_1, T; \sigma)}{k_2 - a_1} \right\}, \\
\sigma_0 &:= \sigma_1 \wedge \sigma_2.
\end{align*}
By $\lim_{\sigma \to 0} c^{BS}(a_2, T; \sigma) = \lim_{\sigma \to 0} \frac{\partial}{\partial k} c^{BS}(a_2, T; \sigma) = 0$ and by 
\begin{align*}
\lim_{\sigma \to 0} \frac{C^2_0 - c^{BS}(a_1,T; \sigma)}{k_2 - a_1} &= \frac{C^2_0 - (S_0 - a_1)}{k_2 - a_1} \\
&= \frac{C^2_0 - S_0}{2(k_2 - a_1)} = \frac{C^2_0 - S_0}{k_2} = D^2_0 \\
&> -1 = \lim_{\sigma \to 0} \frac{\partial}{\partial k} c^{BS}(a_1, T; \sigma),
\end{align*}
$\sigma_0$ is well-defined. Since $c^{BS}(a_2, T; \sigma_0) - \frac{\partial}{\partial k} c^{BS}(a_2, T; \sigma_0) (a_2-k_d) \leq \frac{1}{2} C^d_0$, it follows that
$$
\frac{\partial}{\partial k} c^{BS}(a_2, T; \sigma_0) - \frac{c^{BS}(a_2, T; \sigma_0) - C^d_0}{a_2 - k_d} \geq \frac{C^d_0}{2(a_2 - k_d)} > 0,
$$
and that
$$
\frac{c^{BS}(a_2, T; \sigma_0) - C^d_0}{a_2 - k_d} - D^d_0 = \frac{c^{BS}(a_2, T; \sigma_0)}{a_2 - k_d} > 0.
$$
By $\frac{\partial}{\partial k}c^{BS}(a_1, T; \sigma) \leq \frac{C^2_0 - c^{BS}(a_1, T; \sigma)}{k_2 - a_1}$, $D^2_0 > -1$, $D^i_0 - D^{i-1}_0 > 0$ for $i \in \{3,\ldots,d\}$ and the identities derived above, we can interpolate the points 
$$
(a_1, c^{BS}(a_1, T; \sigma_0)), (k_2/2,c^{BS}(k, T; \sigma_0)), (k_2, C^2_0), \ldots, (k_d, C^d_0), (a_2, c^{BS}_0(a_2; \sigma_0))
$$
and obtain a function $c^{I}(k, t; \sigma_0)$ defined on $A_1 \equiv [a_1, a_2] \times [T/2, T]$ such that
\begin{itemize}
\item $c^{I}(k,t; \sigma_0) = c^{BS}(k,t; \sigma)$ for $(k,t) \in \{a_1, a_2\} \times [T/2, T]$.
\item $\frac{\partial}{\partial k} c^{I}(k, t; \sigma_0) = \frac{\partial}{\partial k} c^{BS}(k, t; \sigma_0)$ for $(k,t) \in \{a_1, a_2\} \times [T/2, T]$.
\item $c^{I}(k_i, T) = C^i_0$ for $i \in \{2, \ldots, d\}$.
\item $c^{I}(k,t; \sigma_0)$ is $C^2$ and convex in $k$ and is $C^1$ and increasing in $t$.
\end{itemize}
Then, define $\widetilde{c}(k,t; \sigma_0)$ on $A_2 \equiv [0, \infty) \times [0,T]$ by
$$
\widetilde{c}(k, t; \sigma_0) =
\begin{cases}
\left(\frac{t-T/2}{T/2} \right)^2 c^{I}(k, t; \sigma_0) + \left(1 - \left(\frac{t-T/2}{T/2} \right)^2 \right) c^{BS}(k, t; \sigma_0), \ & (k,t) \in A_1 \\
c^{BS}(k, t; \sigma_0) & (k,t) \in A_2 \setminus A_1.
\end{cases}
$$
Define
$$
\sigma(x,t) := \frac{1}{k} \sqrt{\frac{2 \frac{\partial}{\partial t} \widetilde{c}(k, t; \sigma_0)}{\frac{\partial^2}{\partial^2 k} \widetilde{c}(k, t; \sigma_0)}}
$$
By Dupire's formula \cite{Dupire:local-vol}, it follows that $\sigma(k,t) = \sigma_0$ on $A_2 \setminus A_1$. Moreover, since $A_1$ is compact and $\sigma(k,t)$ is continuous and positive on $A_1$, we get that $\sigma(k,t) \in [b_1, b_2]$ for some $0 < b_1 < b_2 < \infty$ for $(k,t) \in A_1$. Hence, $\sigma(k,t) \in [l,u]$ on $(k,t) \in A_2$, where $l = b_1 \wedge \sigma_0$ and $u = b_2 \vee \sigma_0$. Then, $S_t := S_0 \cE( \int_0^t \sigma( S_u, u) \, S_u \ud W_u)$ has full support by Proposition \ref{proposition: (FS) property with bounded realised variance}. Furthermore, the corresponding call price function $c_0(k)$ satisfies 
$$
c_0(k_i) = \widetilde{c}(k_i, T; \sigma_0) = c^{I}(k_i, T; \sigma_0) = C^i_0, \quad i \in \{2, \ldots, d\}.
$$
Hence, setting
$$
C^i_t := \E^{\Prob^W}( (S_T - k_i)^+ \, | \, \cF_t), \quad t \in [0,T]
$$
for $i \in \{2,\ldots,d\}$, we conclude that 
$$
\{ (S, C^2,\ldots,C^d); (\cC(\R), \cF^W, \F^W_T, \Prob^W) \} \in \cM(A_0).
$$
and that $\cM(A_0) \neq \emptyset$ if $V^{min}_0 > 0$, $\Delta^{min}_0 > 0$, $D^2_0 > -1$ and $D^d_0 < 0$.

\end{proof}

We can now turn to the proof of Theorem \ref{theorem: regular class of models, call options}. The main ideas are similar to those of the proof of Theorem \ref{theorem: regular class of models, convex claim}. Hence, we do not elaborate on the details as much as in the latter proof.

\begin{proof}[Proof of Theorem \ref{theorem: regular class of models, call options}]

The claim that $\cM(A_0) \neq \emptyset$ iff $V^{min}_0 > 0$, $\Delta^{min}_0 > 0$, $D^2_0 > -1$ and $D^d_0 < 0$ was proved in Proposition \ref{proposition: open set for call options}. It remains to prove that $\cM$ is regular. 

For a given model $\{A ; (\Omega, \cF, \F_T, \Prob) \} \in \cM$, stopping time $\tau \in \cT(\F_T)$ such that $\Prob( \tau < T) > 0$ and matrix $\Sigma \in \mathbb{S}^2_+$, we can construct a model $\{ \widetilde{A} ; (\widetilde{\Omega}, \widetilde{\cF}, \widetilde{\F}_T, \widetilde{\Prob}) \} \in \cM$ and a stopping time $\widetilde{\tau}$ satisfying the conditions in Definition \ref{definition: regular class of models} as follows. 

Define $(\widetilde{\Omega}, \widetilde{\cF},\widetilde{\F}_T,\widetilde{\Prob}) := (\Omega, \cF, \F_T, \Prob) \otimes (\cC(\R), \cF^W, \F^W_T, \Prob^W)$ where $W$ is a Brownian motion and let $m \in \R^{2 \times 2}$ be such that $m^T m = \Sigma$. Denote $z$ to be the function such that $z(A) = V^{min} \wedge \Delta^{min} \wedge (D^2 + 1) \wedge (-D^{d})$. Define
\begin{align*}
&\widetilde{A}_t := A_{t \wedge \tau} + m(W_t - W_{t \wedge \tau}), \quad t \in [\tau, \widetilde{\tau}), \\
& \widetilde{\tau} := \inf \left\{ t \geq \tau: z(\widetilde{A}_t) \leq z(A_\tau)/2 \right\} \wedge \frac{T+\tau}{2},
\end{align*}
and, for $t \in [\widetilde{\tau}, T]$, define
$$
\widetilde{S}_t := \widetilde{S}_{\widetilde{\tau}}  + \int_{\widetilde{\tau}}^t \widetilde{\sigma}(\widetilde{S}_u) S_u \ud W_u, \quad \widetilde{C}^i_t := \E^{\widetilde{\Prob}} \left( (\widetilde{S}_T - k_i)^+ \, \big{|} \, \widetilde{\cF}_t \right), \quad i \in \{2,\ldots,d\},
$$
where $\widetilde{\sigma}(k)$ is an $\widetilde{F}_T$-measurable function bounded in an $\widetilde{F}_\tau$-measurable interval $[l,u] \subset (0,\infty)$ such that
$$
\E^{\widetilde{\Prob}}((\widetilde{S}_T - k_i)^+ \, | \, \widetilde{\cF}_{\widetilde{\tau}} ) = \widetilde{C}^i_{\widetilde{\tau}}, \quad i \in \{2,\ldots,d\}.
$$
The construction of $\widetilde{\sigma}$ can be done in a similar manner as in part (ii) of the proof of Proposition \ref{proposition: open set for convex claim}. Note that the obtained function $\widetilde{\sigma}$ will be measurable with respect to the sigma-algebra generated by the prices $\widetilde{A}_{\widetilde{\tau}}$, hence measurable with respect to $\widetilde{\F}_{\widetilde{\tau}}$.

The construction of $\widetilde{S}$ and $\widetilde{C}$ implies that $\widetilde{A}$ is a $d$-dimensional continuous martingale. Also, it is clear that $\{\tau < \widetilde{\tau} \leq T \} = \{ \tau < T \}$ by construction. The call prices $\widetilde{C}^i$ converge to their payoffs $(\widetilde{S}_T - k_i)^+$ on $\{ \tau < T \}$ since they are equal to conditional expectations of the latter on $[\widetilde{\tau},T]$. They converge to their payoffs on $\{ \tau \geq T \}$ since $\widetilde{A} = A$ on $[0,\tau \wedge T]$ and since $C^i_T = (S_T - k_i)^+$ for $i \in \{2,\ldots,d\}$. As in the proof of Theorem \ref{theorem: regular class of models, convex claim}, the full support property follows by considering the sets (1) $\{ \widetilde{\tau} < T \}$ and (2) $\{\widetilde{\tau} \geq T\}$ and respectively using (1) Proposition \ref{proposition: (FS) property with bounded realised variance} with $\widetilde{\sigma}(x) \in [l,u]$ and (2) $A = \widetilde{A}$ on $[0,\tau \wedge T]$ and the full support property of $A$.

This concludes the proof of $\cM$ being regular.

\end{proof}

% conclusion
\section{Summary}
\label{section: conclusion}

This paper extended the Black-Scholes dynamic hedging paradigm to a market with continuous prices where both the drift and the volatility of these prices is unknown. We first showed that one can obtain a class of model-independent identities by simple application of It\^o's formula along with solutions of coupled PDEs. After discussing some examples of these identities, we showed that they in fact constitute the full set of model-independent strategies with wealth processes of a certain closed form. Applying the general framework to specific settings, we characterised model-independent hedging strategies for two examples of markets: a market with an underlying asset and a convex claim, and a market with an underlying asset and a set of traded calls. The general characterisation result can be interpreted as the characterisation of local martingales of a given closed form and may be of interest outside of the interpretation we discussed within financial markets.

\appendix

\newpage

\numberwithin{equation}{section}
\numberwithin{theorem}{section}

\section{Technical proofs}
\label{appendix: technical proofs}

Within the appendix, we use the shorthand notation $F_i \equiv \frac{\partial}{\partial x_i} F$ and similarly for higher order derivatives of a function $F$. We also write $F \equiv F(x)$ when there is no ambiguity regarding the input $x$ to $F$.

\begin{proof}[Proof of Proposition \ref{proposition: properties of range of X}]

To simplify notation, denote $\cM \equiv \cM(a)$ and $\cR \equiv \cR(X;\cM(a))$, and let $\tau_B$ to be the first hitting time of $X_\cdot \equiv X_\cdot[A]$ to a set $B$ throughout the proof. \\

\item (i) Suppose $\cR$ is not closed. Then there is a sequence $x_n \to x$ such that $x_n \in \cR$ and $x \not\in \cR$. $x \not\in \cR$ implies that $\exists \, \epsilon > 0$ such that for any model in $\cM$, $\tau_{B^c_\epsilon(x)} \geq T$ almost surely. Pick a point $x_n$ such that $d(x_n, x) < \epsilon/2$. By $x_n \in \cR$, there is a model $\{ A ; (\Omega, \cF, \cF_t, \Prob) \} \in \cM$ such that $\Prob(\tau_{B^c_{\epsilon/2}(x_n)} < T) > 0$. But $B^c_{\epsilon/2}(x_n) \subset B^c_{\epsilon}(x)$ implies that $\tau_{B^c_{\epsilon/2}(x_n)} \geq \tau_{B^c_{\epsilon}(x)}$, which yields a contradiction. This proves that $\cR$ is closed. \\

\item (ii) We now prove an intermediate result which will be used in proving connectedness and the claim that $X$ stays within $\cR$ almost surely under all models in $\cM$. We will show that for a closed set $C$, $C \cap \cR = \emptyset$ implies that $\tau_C \geq T$ almost surely under all models in $\cM$. Suppose that this is not true. Then, there is a model $\{ A; (\Omega, \cF, \F_T, \Prob) \} \in \cM$ such that $\Prob( \tau_C < T ) > 0$. By a countable partitioning of $C = \cup_{i=1}^\infty K_i$ where the $K_i$ are compact, it follows that
$$
\sum_{i=1}^\infty \Prob( \tau_{K_i} < T) \geq \Prob \left( \cup_{i=1}^\infty \{ \tau_{K_i} < T \} \right) = \Prob \left( \tau_{\cup_{i=1}^\infty K_i} < T \right) = \Prob( \tau_C > 0) > 0.
$$
Hence, there is a compact set $K \subseteq C$ such that $\Prob( \tau_K < T ) > 0$. Since $K \cap \cR = \emptyset$, for each $x \in K$ there is an $\epsilon(x) > 0$ such that $\Prob \left( \tau_{B^c_{\epsilon(x)}(x)} < T \right) = 0$. By compactness of $K$, there is a finite set $x_1, \ldots, x_k \in K$ such that $K \subseteq \cup_{i=1}^k B^c_{\epsilon(x_i)}(x_i)$. Hence,
$$
\Prob( \tau_K < T ) \leq \Prob \left( \tau_{\cup_{i=1}^k B^c_{\epsilon(x_i)}(x_i)} < T \right) \leq \sum_{i=1}^k \Prob \left( \tau_{B^c_{\epsilon(x_i)}(x_i)} < T \right) = 0,
$$
which yields a contradiction. Hence, $\tau_C \geq T$ almost surely under all models in $\cM$. \\

\item (iii) Suppose that $\cR$ is not connected. Then there are open sets $O_1, O_2 \subset \R^n$ in the relative topology of $\cR$ such that $O_1 \cap O_2 = \emptyset$ and $\cR \subseteq O_1 \cup O_2$. Let $H$ be a hyperplane separating $O_1$ and $O_2$. Without loss of generality, assume that $X_0 \in O_1$. Let $x \in O_2 \cap \cR$. $H \cap \cR = \emptyset$ implies that $d(x,H) > 0$. Let $0 < \epsilon < d(x,H)$. Then, $X_0 \in O_1$ and continuity of paths of $X$ imply that $\tau_H < \tau_{B^c_\epsilon(x)}$. On the other hand, $H \cap \cR = \emptyset$ implies that $\tau_H \geq T$ almost surely under all models in $\cM$ by part (ii) of the proof. Hence, $\tau_{B^c_\epsilon(x)} \geq T$ almost surely under all models in $\cM$, which contradicts $x \in \cR$. Hence, $\cR$ is connected. \\

\item (iv) Let $\{ A; (\Omega, \cF, \F_T, \Prob) \} \in \cM$. For $\epsilon > 0$, define $B^c_\epsilon(\cR) = cl \left( \cup_{x \in \cR} B^c_\epsilon(x) \right)$ and note that $\partial B^c_\epsilon(\cR) \cap \cR = \emptyset$. By part (ii) of the proof, it follows that $\Prob( \tau_{\partial B^c_\epsilon(\cR)} < T) = 0$. Since $\{ X_t \in \cR, \ t \in [0,T) \} =  \cap_{n=1}^\infty \left\{ \tau_{\partial B^c_{1/n}(\cR)} \geq T \right\} 
$, it follows that
$$
\Prob( X_t \in \cR, \ t \in [0,T) ) = 1- \Prob \left( \cup_{n=1}^\infty \left\{ \tau_{\partial B^c_{1/n}(\cR)} < T \right\} \right) \geq 1- \sum_{n=1}^\infty \Prob \left( \tau_{\partial B^c_{1/n}(\cR)} < T \right) = 1,
$$
which concludes the proof.

\end{proof}

\begin{proof}[Proof of Lemma \ref{lemma: (FS) property with stopping times}]

The 'if' statement is trivial (set $\tau = t$). We hence consider the 'only if' statement. Suppose that $S$ has full support, let $I_\tau$ be an $\F_T$-measurable open interval in $\R_+$ and consider a stopping time $\tau \in \cT(\F_T)$.

If $\Prob(\tau < T) = 0$, then there is nothing to prove. If $\Prob(\tau < T) > 0$, consider any $\xi \in \cF_\tau \cap \{ \tau < T \}$ such that $\Prob(\xi) > 0$. $\xi \subseteq \{ \tau < T \}$ implies that 
$$
\lim_{t \to T} \Prob( \xi \cap \{ \tau < t \} ) = P(\xi) > 0.
$$
Hence, $\exists \, t \in (0,T)$ such that $\Prob( \xi \cap \{ \tau < t \} ) > 0$. Note that $\xi \cap \{ \tau < t \} \in \cF_t$ since $\xi \in \cF_\tau$. By the full support property of $S$, $\Prob( S_T \in I_{t \wedge \tau} \, | \, \cF_t) > 0$. It then follows that 
$$
\Prob( \{ S_T \in I_\tau \} \cap \xi ) \geq \Prob( \{ S_T \in I_\tau \} \cap \xi \cap \{ \tau < t \} ) = \Prob( \{ S_T \in I_{t \wedge \tau} \} \cap \xi \cap \{ \tau < t \} ) > 0,
$$
which proves that $\Prob( S_T \in I_\tau \, | \, \cF_\tau ) > 0$ on $\{ \tau < T \}$ as claimed.

\end{proof}

\begin{proof}[Proof of Proposition \ref{proposition: (FS) property with bounded realised variance}]

Without loss of generality, assume that $S_0 = 1$. Let $t \in [0,T)$ and let $I$ be an $\F_T$-adapted open interval. Assume that $\Prob(\tau < T) > 0$, else there is nothing to prove. For any $t \in [0,T)$ and $\xi \in \cF_t$ such that $\Prob( \xi \cap \{ \tau < T \} ) > 0$, we need to show that $\Prob( \{  S_T \in I_t \} \cap \xi \cap \{ \tau < T\}) > 0$. Without loss of generality, assume that $\tau \leq T$ (else use $\tau \wedge T$ in the proof).

Denoting $\F_\infty$, consider a filtered probability space $(\Omega^\prime, \cF^\prime, \F^\prime_\infty, \Prob^\prime)$ containing $(\Omega, \cF, \F_\infty, \Prob)$ such that $W^\prime: \Omega^\prime \times [0,\infty) \to \R$ is a Brownian motion independent of $\cF$ defined on it. Define 
\begin{align*}
Y_\cdot &:=  \int_{\cdot \wedge t}^\cdot \sigma_u \ud W_u, \\
Y^\prime_\cdot &:= Y_{\cdot \wedge t} + W^\prime_{\langle Y \rangle_\cdot} - W^\prime_{\langle Y \rangle_{\cdot \wedge t}}, \\
S^\prime_\cdot &:= S_{\cdot \wedge t} \cE(Y^\prime).
\end{align*}
Note that $Y = 0$ and $S^\prime = S$ on $[0, t]$. By a standard time change argument (with the starting filtration being $\cF_t$ instead of the usual trivial filtration), the laws of $S_T$ and $S^\prime_T$ conditional on $\cF_t \cap \{ \tau < T \}$ coincide. Let $\xi \in \cF_t$. Then,
$$
\Prob( \{ S_T \in I_t \} \cap \xi \cap \{ \tau < T \} ) = \Prob^\prime( S_T \in I_t \} \cap \xi \cap \{ \tau < T \} ) = \Prob^\prime( \{ S^\prime_T \in I_t \} \cap \xi \cap \{ \tau < T \}).
$$
Denoting $\log I + c \equiv (\log(a) + c, \log(b) + c)$ for $I = (a,b)$, define the $\F_T$-adapted open interval
$$
I^\prime := \log I - \log S.
$$
Note that $Y^\prime = \log S^\prime + \frac{1}{2} \langle Y^\prime \rangle$ on $[t, T]$, and denote $t^\prime \equiv t \vee \tau$ and $\langle Y \rangle_{t^\prime,T} \equiv \langle Y \rangle_T - \langle Y \rangle_{t^\prime}$. Then,
\begin{align*}
&\Prob( \{ S_T \in I_{t^\prime} \} \cap \xi \cap \{ \tau < T \} ) \\
&= \Prob^\prime \left( \{ \log S^\prime_T - \log S^\prime_{t^\prime} + \frac{1}{2} \langle Y \rangle_{t^\prime,T} \in \log I_{t^\prime} - \log S^\prime_{t^\prime} + \frac{1}{2} \langle Y \rangle_{t^\prime,T} \} \cap \xi \cap \{ \tau < T \} \right) \\
&= \Prob^\prime \left( \{ Y^\prime_T - Y^\prime_{t^\prime} \in I^\prime_{t^\prime} + \frac{1}{2} \langle Y \rangle_{{t^\prime},T}  \} \cap \xi \cap \{ \tau < T \} \right) \\
&= \Prob^\prime \left( \left\{ W^\prime_{\langle Y \rangle_{t^\prime}} + \langle Y \rangle_{t^\prime,T} - W^\prime_{\langle Y \rangle_{t^\prime}} \in I^\prime_{t^\prime} + \frac{1}{2} \langle Y \rangle_{t^\prime,T} \right\} \cap \xi \cap \{ \tau < T \} \right) \\
&\geq \Prob^\prime \left( \left\{ W^\prime_{\langle Y \rangle_{t^\prime} + \langle Y \rangle_{t^\prime,T}} - W^\prime_{\langle Y \rangle_{t^\prime}} \in I^\prime_{t^\prime} + \frac{1}{2} \langle Y \rangle_{t^\prime,T} \right\} \cap \left\{ l({t^\prime}) \leq \langle Y \rangle_{t^\prime,T} \leq u({t^\prime}) \right\} \cap \xi \cap \{ \tau < T \} \right) \\
&=\E \left( \Prob^\prime \left( W^\prime_{\langle Y \rangle_{t^\prime} + \langle Y \rangle_{t^\prime,T}} - W^\prime_{\langle Y \rangle_{t^\prime}} \in I^\prime_{t^\prime} + \frac{1}{2} \langle Y \rangle_{{t^\prime},T}  \, \big{|} \, \langle Y \rangle_{t^\prime,T}, \tau < T \right) \I_{\{ l({t^\prime}) \leq \langle Y \rangle_{t^\prime,T} \leq u({t^\prime}) \} \cap \xi \cap \{ \tau < T \}} \right)
\end{align*}
Note that 
$$
\Prob^\prime \left( W^\prime_{\langle Y \rangle_{t^\prime} + \langle Y \rangle_{t^\prime,T}} - W^\prime_{\langle Y \rangle_{t^\prime}} \in I^\prime_{t^\prime} + \frac{1}{2} \langle Y \rangle_{{t^\prime},T}  \, \big{|} \, \langle Y \rangle_{{t^\prime},T}, \tau < T \right) > 0
$$
on $\{ l({t^\prime}) \leq \langle Y \rangle_{{t^\prime},T} \leq u({t^\prime}) \}$ by independence of $W^\prime$ from $\cF$. Also, 
$$
\Prob^\prime \left( \{ l(t^\prime) \leq \langle Y \rangle_{t^\prime,T} \leq u(t^\prime) \} \cap \xi \cap \{ \tau < T \} \right) > 0
$$ 
by $\Prob( l(t^\prime) \leq \int_{t^\prime}^T \sigma_u^2 \ud u \leq u(t^\prime) \, | \, \cF_{t^\prime}) > 0$ on $\{ \tau < T \}$ and by $\Prob(\xi \cap \{ \tau < T \}) > 0$. Hence,
$$
\Prob( \{ S_T \in I_{t^\prime} \} \cap \xi \cap \{ \tau < T \} ) > 0,
$$
which concludes the proof of $\Prob( S_T \in I_{t^\prime} \, | \, \cF_{t^\prime} ) > 0$ on $\{ \tau < T \}$.

\end{proof}

\section{Equivalence Results}
\label{appendix: equivalence}

One can generalise Corollary \ref{corollary: claims on square drawdown and realised variance} by adding more components to $X$. In particular, for $X = (S, D^M, Q^w)$, $\cS_X(\cD)$ is the set of $C^{2,2,1}$ functions which are solutions to
\begin{equation}
\label{PDE for price, squared drawdown and realised variance}
w \frac{\partial}{\partial x_3} F + \frac{\partial}{\partial x_2} F + \frac{1}{2} \frac{\partial^2}{\partial x_1^2}F + 2 x_2 \frac{\partial^2}{\partial x_2^2} F - 2 \sqrt{x_2} \frac{\partial^2}{\partial x_1 \partial x_2} F = 0,
\end{equation}
subject to some boundary conditions. Corollary \ref{corollary: model-independent replication} then implies that one can generate portfolios with wealth process $F(X^B_\cdot)$ by starting with capital $F(X_0)$ and holding a position $H = \frac{\partial}{\partial x_1} F(X) - 2 \sqrt{D^M} \frac{\partial}{\partial x_2} F(X)$ in $S$. 

This choice of $X$ yields similar hedging strategies as Proposition \ref{proposition: claims on S, M and Q^w} involving mixed boundary PDEs. Since $(S, D^M, Q^w)$ and $(S, M, Q^w)$ are homeomorphic, in fact diffeomorphic except on $\{ D^M_t = 0 \} = \{ S_t = M_t \}$, one may expect the sets of self-financing strategies in $S$ having wealth processes equal to $F^D(S_\cdot, D^M_\cdot, Q^w_\cdot)$ for solutions $F^D$ to \eqref{PDE for price, squared drawdown and realised variance} to be equivalent to strategies with wealth processes of the form $F^M(S_\cdot, M_\cdot, Q^w_\cdot)$ for solutions $F^M$ to \eqref{running maximum PDE} -- \eqref{mixed BC}, up to regularity conditions on $F^D$ and $F^M$. The next proposition shows this for $w = 1$. We use the following notational convention within the proposition and its proof:
\begin{itemize}
\item $w = 1$ when referring to \eqref{PDE for price, squared drawdown and realised variance} and \eqref{running maximum PDE}.
\item For a set $K$ and a function $g$ defined on it, define $g(K) := \{ g(x): x \in K \}$. 
\end{itemize}

\begin{proposition}
\label{proposition: equivalence between F^M and F^D}
Define $y(x) := (x_1, (x_2 - x_1)^2, x_3)$ and $\ytilde(x) := (x_1, x_1 + \sqrt{x_2}, x_3)$. \\
Consider a connected set $\cD \subset \R^3$ with smooth boundary and a solution $F^D \in C^{2,2,1}(\cD)$ to \eqref{PDE for price, squared drawdown and realised variance}. Then, the function
\begin{equation}
\label{F^D to F^M}
F^M := F^D \circ y
\end{equation}
is a $C^{2,1,1}$ solution to \eqref{running maximum PDE} on $\ytilde(\cD)$. Furthermore, \eqref{mixed BC} holds on $\ytilde(\cD)$ provided $\frac{\partial}{\partial y_2} F^D(y_1, y_2, y_3)$ is well-defined on $\cD \cap \{y_2 = 0\}$. \\
Conversely, consider a connected set $\cDtilde$ with smooth boundary and a solution $\Ftilde^M \in C^{2,2,1}(\cDtilde)$ to \eqref{running maximum PDE} -- \eqref{mixed BC}. Then, the function
$$
\Ftilde^D := \Ftilde^M \circ \ytilde
$$
is a $C^{2,2,1}$ solution to \eqref{PDE for price, squared drawdown and realised variance} on $\ytilde(\cDtilde) \setminus \{x_2 = 0\}$. 
\end{proposition}

Since $D^M$ is equal to zero every time $S$ attains its running maximum, it is important that $\Ftilde^D(x_1,x_2,x_3)$ be differentiable from the right in the $x_2$ variable on $\{x_2 = 0\} \cap y(\cDtilde)$. If one solves for $\Ftilde^D$ via $\Ftilde^M$, as in the proposition, one may check for differentiability of $\Ftilde^D$ after applying the change in variables, or one may refer to Lemma \ref{lemma: sufficient conditions for differentiability of F^D} in Appendix \ref{appendix: appendix} for sufficient conditions on $\Ftilde^M$ for differentiability of $\Ftilde^D$ at $\{x_2 = 0\} \cap y(\cDtilde)$.

\begin{proof}

Let $F^D$ be a solution to \eqref{PDE for price, squared drawdown and realised variance} and $F^M$ be defined by \eqref{F^D to F^M}. Using the shorthand notation $y \equiv y(x)$, direct computation gives
\begin{align*}
F^M_1(x) &= F^D_1(y) - 2(x_2 - x_1) F^D_2(y) \\
F^M_{11}(x) &= F^D_{11}(y) + 2 F^D_2(y) - 4 (x_2 - x_1) F^D_{12}(y) + 4 (x_2 - x_1)^2 F^D_{22}(y) \\
F^M_{2}(x) &= 2 (x_2 - x_1) F^D_2(y) \\
F^M_3(x) &= F^D_3(y).
\end{align*}
for $x \in \ytilde(\cD)$. Hence, $F^M_2 = 0$ on $\{x_1 = x_2\} \cap \ytilde(D)$ provided $F^D_2(y_1, y_2, y_3)$ is well-defined on $\cD \cap \{ y_2 = 0 \}$, and
\begin{align*}
F^M_3(x) + \frac{1}{2} F^M_{11}(x) &= F^D_3(y) + \frac{1}{2} F^D_{11}(y) + F^D_2(y) - (x_2 - x_1) F^D_{12}(y) + 2 (x_2 - x_1)^2 F^D_{22}(y) \\
&=   F^D_3(y) + \frac{1}{2} F^D_{11}(y) + F^D_2(y) - \sqrt{y_2} F^D_{12}(y) + 2 y_2 F^D_{22}(y) = 0
\end{align*}
for $x \in \ytilde(D)$. Conversely, let $\Ftilde^M$ be a solution to \eqref{running maximum PDE} -- \eqref{mixed BC} on $\cDtilde$. Then, denoting $\ytilde = \ytilde(x)$, direct computation yields
\begin{align*}
\Ftilde^D_1(x) &= \Ftilde^M_1(\ytilde) + \Ftilde^M_2(\ytilde) \\
\Ftilde^D_{11}(x) &= \Ftilde^M_{11}(\ytilde) + 2 F^M_{12}(\ytilde) + \Ftilde^M_{22}(\ytilde) \\
\Ftilde^D_2(x) &= \frac{1}{2 \sqrt{x_2}} \Ftilde^M_2(\ytilde) \\
\Ftilde^D_{12}(x) &= \frac{1}{2 \sqrt{x_2}} \left(\Ftilde^M_{12}(\ytilde) + \Ftilde^M_{22}(\ytilde) \right) \\
\Ftilde^D_{22}(x) &= - \frac{1}{4 (x_2)^{3/2}} \Ftilde^M_2(\ytilde) + \frac{1}{4 x_2} \Ftilde^M_{22}(\ytilde) \\
\Ftilde^D_3(x) &= \Ftilde^M_3(\ytilde)
\end{align*}
for $x \in y({\cDtilde}) \setminus \{x_2 = 0\}$. Hence, for such $x$,
\begin{align*}
&\Ftilde^D_3(x) + \Ftilde^D_2(x) + \frac{1}{2} \Ftilde^D_{11}(x) + 2 x_2 \Ftilde^D_{22}(x) - 2 \sqrt{x_2} \Ftilde^D_{12}(x) \\
& = \Ftilde^M_3(\ytilde) + \frac{1}{2 \sqrt{x_2}} \Ftilde^M_2 (\ytilde) + \frac{1}{2} \left( \Ftilde^M_{11}(\ytilde) + 2 \Ftilde^M_{12}(\ytilde) + \Ftilde^M_{22}(\ytilde) \right) \\
&\quad\quad\quad + 2 x_2 \left( - \frac{1}{4 (x_2)^{3/2}} \Ftilde^M_2(\ytilde) + \frac{1}{4 x_2} \Ftilde^M_{22}(\ytilde) \right) - 2 \sqrt{x_2} \left( \frac{1}{2 \sqrt{x_2}} \left(\Ftilde^M_{12}(\ytilde) + \Ftilde^M_{22}(\ytilde) \right) \right) \\
&= \Ftilde^M_3(\ytilde) + \frac{1}{2} \Ftilde^M_{11}(\ytilde) = 0.
\end{align*}

\end{proof}

We now derive sufficient conditions of second order right-differentiability of $\Ftilde^D$ at $x_2 = 0$.

\begin{lemma}
\label{lemma: sufficient conditions for differentiability of F^D}

Consider $\Ftilde^M$, $\Ftilde^D$, $y$, $\ytilde$ and $\cDtilde$ as in Proposition \ref{proposition: equivalence between F^M and F^D}. If $\Ftilde^M \in C^{2,5,1}(\cDtilde \cap \{x_1 = x_2\})$ and $\frac{\partial^3}{\partial x_2^3} \Ftilde^M = 0$ on $\cDtilde \cap \{x_1 = x_2\}$, then there exists a well-defined right derivative of $\Ftilde^D$ at $x_2 = 0$, denoted by
$$
\Ftilde^D_{2+}(x_1,0,x_3) := \lim_{\epsilon \to 0} \frac{\Ftilde^D(x_1, \epsilon, x_3) - \Ftilde^D(x_1, 0, x_3)}{\epsilon},
$$
and it is equal to $\lim_{x_2 \to 0} \Ftilde^D_2(x_1, x_2, x_3)$ for $(x_1, 0, x_3) \in y(\cDtilde)$. 

\end{lemma}

In the proof below, recall that $w = 1$ when referring to  \eqref{running maximum PDE} and \eqref{PDE for price, squared drawdown and realised variance}.

\begin{proof}[Proof of lemma]

By the definition of $\Ftilde^D$,
\begin{equation}
\label{F^D one-sided derivative in x_2}
\Ftilde^D_{2+}(x_1, 0, x_3) = \lim_{\epsilon \to 0} \frac{\Ftilde^M(x_1, x_1 + \sqrt{\epsilon}, x_3) - \Ftilde^M(x_1, x_1, x_3)}{\epsilon}.
\end{equation}
Since $\Ftilde^M \in C^{2,2,1}(\cDtilde)$ and $\Ftilde^M_{2}(x_1,x_1,x_3) = 0$ for $(x_1, 0, x_3) \in y(\cDtilde)$, a second order Taylor expansion of $\Ftilde^M$ around $(x_1, x_1, x_3)$ yields
$$
\Ftilde^M(x_1, x_1 + \sqrt{\epsilon}, x_3) - \Ftilde^M(x_1, x_1, x_3) = \frac{1}{2} \epsilon \Ftilde^M_{22}(x_1,x_1+\xi(\epsilon),x_3)
$$
for some $\xi(\epsilon) \in [0, \sqrt{\epsilon}]$. Along with \eqref{F^D one-sided derivative in x_2} and the continuity of $F^M_{22}$, this gives
$$
\Ftilde^D_{2+}(x_1, 0, x_3) = \lim_{\epsilon \to 0} \frac{1}{2} \Ftilde^M_{22}(x_1, x_1 + \xi(\epsilon), x_3) = \frac{1}{2} \Ftilde^M_{22}(x_1, x_1, x_3)
$$
for $(x_1, 0, x_3) \in y(\cDtilde)$. 
To conclude the proof, we need to show second order continuous right-hand differentiability of $\Ftilde^D$ in $x_2$ provided the assumption that $\Ftilde^M$ is $C^{2,5,1}(\cDtilde \cap \{ x_1 = x_2 \})$ and $\frac{\partial^3\Ftilde^M}{\partial x_2^3} = 0$ on $\cDtilde \cap \{ x_1 = x_2 \})$. In particular, we need to show that 
$$
\Ftilde^D_{22+}(x_1,0,x_3) := \lim_{\epsilon \to 0} \frac{\Ftilde^D_2(x_1, \epsilon, x_3) - \Ftilde^D_{2+}(x_1, 0, x_3)}{\epsilon}
$$
is well-defined and coincides with $\lim_{x_2 \to 0} \Ftilde^D_{22}(x_1, x_2, x_3)$ for $(x_1, 0 , x_3) \in y(\cDtilde)$. First, note that
\begin{equation}
\label{F^D one-sided second derivative in x_2}
\Ftilde^D_{22+}(x_1, 0, x_3) = \lim_{\epsilon \to 0} \frac{\frac{1}{2 \sqrt{\epsilon}} \Ftilde^M_2(x_1, x_1 + \sqrt{\epsilon}, x_3) - \frac{1}{2} \Ftilde^{M}_{22}(x_1, x_1, x_3)}{\epsilon}.
\end{equation}
By $\Ftilde^M \in C^{2,5,1}(\cDtilde \cap \{x_1 = x_2\})$ and $\frac{\partial^3}{\partial x_2^3} \Ftilde^M(x_1, x_1, x_3) = 0$, a third order Taylor expansion of $\Ftilde^M_2$ yields
$$
\Ftilde^M_2(x_1, x_1 + \sqrt{\epsilon}, x_3) - \sqrt{\epsilon} \Ftilde^M_{22}(x_1, x_1, x_3) = \frac{1}{6} \epsilon^{3/2} \frac{\partial^4}{\partial x_2^4} \Ftilde^M(x_1, x_1 + \gamma(\epsilon), x_3)
$$
for some $\gamma(\epsilon) \in [0, \sqrt{\epsilon}]$. Hence, \eqref{F^D one-sided second derivative in x_2} and the continuity of $\frac{\partial^4}{\partial x_2^4} \Ftilde^M$ at $(x_1, x_1, x_3)$ imply
$$
\Ftilde^D_{22+}(x_1, 0, x_3) = \lim_{\epsilon \to 0}  \left( \frac{1}{12} \frac{\partial^4}{\partial x_2^4} \Ftilde^M(x_1, x_1 + \xi(\epsilon), x_3) \right) =  \frac{1}{12} \frac{\partial^4}{\partial x_2^4} \Ftilde^M(x_1, x_1, x_3).
$$
On the other hand, for fixed $(x_1, 0, x_3) \in y(\cDtilde)$, define $G( \, \cdot \, ; x_1, x_3): \R_+ \to \R$ by
$$
G(z \,; x_1, x_3) =  z \left( \Ftilde^M_{22}(x_1, x_1 + z, x_3) \right)  - \Ftilde^M_{2}(x_1, x_1 + z, x_3)
$$
Then, by $\Ftilde^M \in C^{2,5,1}$,
\begin{align*}
G^\prime(z \,; x_1, x_3) &= z \left( \frac{\partial^3}{\partial x_2^3} \Ftilde^M(x_1, x_1 + z, x_3) \right), \\
G^{\prime\prime}(z \,; x_1, x_3) &=  z \left( \frac{\partial^4}{\partial x_2^4} \Ftilde^M(x_1, x_1 + z, x_3) \right) + \frac{\partial^3}{\partial x_2^3} \Ftilde^M(x_1, x_1 + z, x_3), \\
G^{\prime\prime\prime}(z \,; x_1, x_3) &=  z \left( \frac{\partial^5}{\partial x_2^5} \Ftilde^M(x_1, x_1 + z, x_3) \right) + 2 \frac{\partial^4}{\partial x_2^4} \Ftilde^M(x_1, x_1 + z, x_3).
\end{align*}
Hence,
\begin{align*}
\lim_{x_2 \to 0} \Ftilde^D_{22}(x_1, x_2, x_3) &= \lim_{x_2 \to 0} \frac{- \Ftilde^M_2(x_1,x_1 + \sqrt{x_2}, x_3) + \sqrt{x_2} \Ftilde^M_{22}(x_1, x_1 + \sqrt{x_2}, x_3)}{4 (x_2)^{3/2}} \\
&= \lim_{x_2 \to 0} \frac{ G\left( \sqrt{x_2} \,; x_1, x_3 \right) }{4 (x_2)^{3/2}}.
\end{align*}
Note that $G(0 \,; x_1, x_3) = G^{\prime}(0 \,; x_1, x_3) = G^{\prime\prime}(0 \,; x_1, x_3) = 0$ since $\frac{\partial^i\Ftilde^M}{\partial x_2^i}(x_1, x_1, x_3) = 0$ for $i = 1$ and $i = 3$. Hence, a third order Taylor expansion of $G$ around $0$ yields
\begin{align*}
&G\left( \sqrt{x_2} \,; x_1, x_3 \right) \\
&= \frac{1}{6} (x_2)^{3/2} \left( \sqrt{x_2} \left( \frac{\partial^5}{\partial x_2^5} \Ftilde^M(x_1, x_1 + \rho(x_2), x_3) \right) + 2 \frac{\partial^4}{\partial x_2^4} \Ftilde^M(x_1, x_1 + \rho(x_2), x_3) \right)
\end{align*}
for some $\rho(x_2) \in [0,\sqrt{x_2}]$. Since $\frac{\partial^4}{\partial x_2^4} \Ftilde^M$ is continuous at $(x_1, x_1, x_3)$,
\begin{align*}
&\lim_{x_2 \to 0} \Ftilde^D_{22}(x_1, x_2, x_3) \\
&= \lim_{x_2 \to 0} \left( \frac{1}{24} \left( \sqrt{x_2} \left( \frac{\partial^5}{\partial x_2^5} \Ftilde^M(x_1, x_1 + \rho(x_2), x_3) \right) + 2 \frac{\partial^4}{\partial x_2^4} \Ftilde^M(x_1, x_1 + \rho(x_2), x_3) \right) \right) \\
&= \frac{1}{12} \frac{\partial^4}{\partial x_2^4} \Ftilde^M(x_1, x_1, x_3) \\
&= \Ftilde^D_{22+}(x_1, 0, x_3).
\end{align*}

\end{proof}

The following proposition shows that there is a stochastic solution to \eqref{running maximum PDE} -- \eqref{mixed BC} with boundary condition 
\begin{equation}
\label{running maximum BC}
F(x_1, x_2, q) = f(x_1, x_2),
\end{equation}
provided an implicit regularity condition on $f$. We do not claim this to be a contribution to the literature on stochastic solutions of parabolic PDEs, but provide the result for completeness. 

\begin{proposition}
\label{proposition: stochastic solution for F(S,M,[S])}

Let $W$ be a Brownian motion on the canonical space $(\Omega, \cF, \F, \Prob) \equiv (\cC(\R), \cF^W, \F^W, \Prob^W)$. Let $M \equiv M^W$ be the running maximum of $W$. Suppose that $f: \R_+^2 \to \R$ is such that
$$
F(x_1,x_2,x_3) := \E \left( f \left( x_1 + W_{q - x_3}, x_2 \vee (x_1 + M_{q- x_3}) \right) \right), \quad  0 \leq x_1 \leq x_2, \ \ 0 \leq x_3 \leq q.
$$
is $C^{2,1,1}$ on $\cD = \{x_1 \leq x_2 \} \cap \{ 0 \leq x_3 < q \} \subset \R^3$ and $\int_0^\cdot F_1(W_t,M_t,t) \ud W_t$ is a martingale. Then for $0 \leq t \leq q$,
\begin{equation}
\label{lookback Markovian property}
\E \left( f(x + W_q, x + M_q) \ | \ \cF_t \right) =  F(W_t, M_t, t) \quad \Prob-a.e.,
\end{equation}
and $F(x)$ is a solution to \eqref{running maximum PDE} -- \eqref{mixed BC} with boundary condition  \eqref{running maximum BC}.

\end{proposition}

\begin{proof}

The claim that $F(x)$ defined by \eqref{stochastic solution for F} satisfies \eqref{lookback Markovian property} follows from the strong Markov property of Brownian motion (Ob\l\'oj and Yor \cite{Obloj-Yor}). The boundary condition \eqref{running maximum BC} follows from \eqref{lookback Markovian property} with $t=q$. To show that $F(x)$ is a solution to \eqref{running maximum PDE} provided the assumed regularity conditions on $F$, we resort to standard martingale arguments.

Since $M$ is continuous and of finite variation, and since $F$ was assumed to be $C^{2,1,1}$, It\^o's formula gives
\begin{align}
\ud F(W_t, M_t, t) &= F_1 \ud W_t + F_2 \ud M_t + F_3 \ud t + \frac{1}{2} F_{11} \ud t \nonumber \\ 
& = F_1 \ud W_t + F_2 \ud M_t + \left( F_3 + F_{11} \right) \ud t.\label{lookback solution proof eqn 2} 
\end{align}
Formally defining $\int_0^\cdot \I_{\{ W_u = M_u \}} H_u \ud W_u := \lim_{\epsilon \downarrow 0} \int_0^\cdot I_{\{ W_u > M_u - \epsilon \}} H_u \ud W_u$, this implies that for $0 \leq s < t \leq q$,
\begin{align*}
&\int_s^t \I_{\{W_u = M_u\}} dF(W_u, M_u, u) \\
&= \int_s^t  \I_{\{W_u = M_u\}} \ F_1 \ud W_u + \int_s^t \I_{\{W_u = M_u\}} F_2 \ud M_u + \int_s^t \I_{\{W_u = M_u\}} \left( F_3 + \frac{1}{2} F_{11} \right) \ud u. 
\end{align*}
The set $\{ u \in [s,t] : W_u = M_u \}$ has Lebesgue measure zero $\Prob$-a.e., the third term on the
second line disappears. Since $M$ is carried by the set $\{ W = M \}$, it follows that
$$
F_2 \ud M_u = F_2 \I_{\{ W_u = M_u \}} \ud M_u = \I_{\{ W_u = M_u \}} dF(W_u, M_u, u) - \I_{\{ W_u = M_u \}} \ F_1(Z_u) \ud W_u. 
$$
Since $\int_0^\cdot F_1(W_t, M_t, t) \ud W_t$ is a martingale by assumption and $\{ F(W_t, M_t, t) \}_{t \in [0,T]}$ is a martingale by
\eqref{lookback Markovian property},
\begin{align*}
0 &\leq \E\left( \int_s^t \I_{\{ F_2(W_u, M_u, u) > 0 \}} F_2 \ud M_u \right) \\
&= \E \bigg{(} \int_s^t \I_{ \{ F_2(W_u, M_u, u) > 0 \} \cap \{ W_u = M_u \} } dF(W_u, M_u, u) \\
& \hspace{100pt} + \int_s^t \I_{ \{F_2(W_u, M_u, u) > 0 \} \cap \{ W_u = M_u \} } F_1 \ud W_u \bigg{)} = 0.
\end{align*}
Hence,
\begin{equation}
\label{lookback solution proof eqn 3}
\E \left( \int_s^t \I_{\{ F_2(W_u, M_u, u) > 0 \}} F_2 \ud M_u \right) = 0.
\end{equation}
By a symmetric argument,
\begin{equation}
\label{lookback solution proof eqn 4}
\E \left( \int_s^t \I_{\{ F_2(W_u, M_u, u) < 0 \}} F_2 \ud M_u \right) = 0.
\end{equation}
Combining \eqref{lookback solution proof eqn 3} and \eqref{lookback solution proof eqn 4} yields $\E \left( \int_s^t  | F_2(W_u, M_u, u) | \ud M_u \right) = 0$. Since $M$ is non-decreasing, this implies $\int_s^t  | F_2(W_u, M_u, u) | \ud M_u = 0$ $\Prob$-a.e.. Recalling that $W = M$ on points of increase of $M$, we get 
\begin{equation}
\label{lookback solution proof eqn 5}
\int_s^t  | F_2(M_u, M_u, u) | \ud M_u = 0, \quad \quad \Prob-a.e.
\end{equation}
By assumption, $\widehat{F}(x,t) := |F_2(x,x,t)|$ is continuous, hence absolute continuity of the Lebesgue measure with respect to the law of $M_u$ on $[0,\infty)$ for all $u \in (0,q]$ implies that $F_2(x,x,y) = 0$ for $(x,y) \in \R_+ \times (0,q]$, which yields \eqref{mixed BC}. Combining this with \eqref{lookback solution proof eqn 2} and the martingale property of $W$ and of $\{ F(W_t,M_t,t) \}_{t \in [0,q]}$, it follows that $\int_0^t ( F_3(W_u, M_u, u) + \frac{1}{2} F_{11}(W_u, M_u, u) ) \ud u$ is a local martingale starting at $0$. Since $F_3(x) +\frac{1}{2} F_{11}(x)$ is
continuous by assumption, a standard stopping time argument yields \eqref{running maximum PDE}.

\end{proof}

\section*{Acknowledgements}

I would like to thank my supervisor, Prof. Michael Monoyios, for his continuous support throughout my DPhil work. I thank Prof. Monoyios in particular for his patience, for his openness to exploring various research ideas which were of interest to me, as well as for his help in formulating research results coherently and concisely in writing.

I thank Prof. Jan Ob\l\'oj for his insightful discussions regarding questions I have had regarding martingale theory, robust methods in finance, and other areas of research related to financial mathematics. I am also very grateful to Prof. Johannes Ruf, Prof. Sam Cohen and Dr. Pietro Siorpaes who have been very supportive and open to discuss many technical questions I have come across during the DPhil research.

I thank the many fellow students with whom I have dicsussed relevant research topics, notably Vladimir Cherny, Philippe Charmoy, Peter Spoida, Zhaoxu Hou, Sergey Shahverdyan, Sean Ledger, Matthieu Mariapragassam, Arnaud Lionnet, Yves-Laurent Kom Samo, Guy Flint and many others.

I am grateful to the Oxford-Man Institute (OMI) for funding my research as well as for the facilities and resources they have provided. I thank Prof. Marek Musiela for several lengthy discussions regarding some research topics I was interested in. I also thank Prof. Terry Lyons for informal high level discussions regarding the general intuition behind the role of pathwise integration theory in financial mathematics. Finally, I thank Jane Williams and the OMI administration and IT teams for providing the best administrative and technical support one can ask for as a graduate student.

I would also like to thank two practitioners in industry. First of all, I thank Dr. Peter Carr for taking the time to answer some questions I have had regarding his research, as well as for discussing research ideas which have led to some of the material presented in this thesis. Secondly, Dr. George Lowther's blog almostsure.wordpress.com has been of tremendous value in understanding the intuition behind many technical results in stochastic analysis. Dr. Lowther has furthermore been very helpful and responsive to questions I have asked about the content of his posts on several occasions.

The most important acknowledgement goes to my family who have always pushed me to high standards of achievement and whose continuous support has been invaluable.

%next line adds the Bibliography to the contents page
\addcontentsline{toc}{chapter}{Bibliography}
\bibliography{refs}
\bibliographystyle{plain}

\end{document}